\pdfoutput=1
\documentclass[10pt,pra,aps,superscriptaddress,nofootinbib,twocolumn,longbibliography]{revtex4-1}
\usepackage{amsmath,bbm}
\usepackage{amsthm}
\usepackage{amsmath}
\usepackage{latexsym}
\usepackage{amssymb}
\usepackage{float}
\usepackage{graphicx}           
\usepackage{color}
\usepackage{xcolor}
\usepackage{mathpazo}
\usepackage{comment}
\usepackage{enumitem}
\usepackage{multirow}
\usepackage{setspace}
\usepackage{graphicx}
\usepackage{caption}
\usepackage{subcaption}

\usepackage[colorlinks=true,linkcolor=blue,citecolor=magenta,urlcolor=blue]{hyperref}
\usepackage{svg}

\newcommand{\be}{\begin{equation}}
\newcommand{\ee}{\end{equation}}
\newcommand{\bea}{\begin{eqnarray}}
\newcommand{\eea}{\end{eqnarray}}

\def\squareforqed{\hbox{\rlap{$\sqcap$}$\sqcup$}}
\def\qed{\ifmmode\squareforqed\else{\unskip\nobreak\hfil
\penalty50\hskip1em\null\nobreak\hfil\squareforqed
\parfillskip=0pt\finalhyphendemerits=0\endgraf}\fi}
\def\endenv{\ifmmode\;\else{\unskip\nobreak\hfil
\penalty50\hskip1em\null\nobreak\hfil\;
\parfillskip=0pt\finalhyphendemerits=0\endgraf}\fi}

\newcommand{\tr}{\text{Tr}}

\newcommand{\w}{\omega}

\newcommand{\ket}[1]{|#1\rangle}
\newcommand{\bra}[1]{\langle#1|}

\newcommand{\la}{\langle}
\newcommand{\ra}{\rangle}

\newcommand{\blk}{\color{black}}

\makeatletter
\newtheorem*{rep@theorem}{\rep@title}
\newcommand{\newreptheorem}[2]{%
\newenvironment{rep#1}[1]{%
 \def\rep@title{#2 \ref{##1}}%
 \begin{rep@theorem}}%
 {\end{rep@theorem}}}
\makeatother

\newtheorem{thm}{Theorem}
\newreptheorem{thm}{Theorem}
\newtheorem{lemma}{Lemma}
\newtheorem{definition}{Definition}

\newtheorem{coro}{Corollary}

\newtheorem{example}{Example}

\begin{document}


\title{Failure of epistemic models to explain the antidistinguishability of quantum mixed preparations}



\author{Sagnik Ray}
\affiliation{School of Physics, Indian Institute of Science Education and Research Thiruvananthapuram, Kerala 695551, India}
\author{Visweshwaran R}
\affiliation{School of Physics, Indian Institute of Science Education and Research Thiruvananthapuram, Kerala 695551, India}
\author{Debashis Saha}
\affiliation{School of Physics, Indian Institute of Science Education and Research Thiruvananthapuram, Kerala 695551, India}
\affiliation{Department of Physics, School of Basic Sciences, Indian Institute of Technology Bhubaneswar, Odisha 752050, India}


\begin{abstract}
 We investigate the limits of epistemic models in reproducing the empirical predictions of general quantum preparations. This involves comparing the \textit{common quantum overlap} determined by the anti-distinguishability of a set of preparations with the \textit{common epistemic overlap} of the probability distribution over the ontic states describing these preparations. We show that no epistemic model with non-zero common overlap can reproduce the predictions of quantum theory for mixed preparations, even for qubit systems. We explicitly provide sets of distinct mixed preparations, starting from qutrit systems, that give rise to identical quantum mixed states for which the common epistemic overlap must be zero. Finally, we demonstrate the strongest form of \textit{preparation contextuality} by presenting pairs of mixed preparations that yield identical quantum mixed states while their epistemic overlap must vanish.
\end{abstract}

\maketitle


\section{Introduction} 

One of the central unresolved queries within quantum theory revolves around the existence of an underlying theory where the quantum description of physical systems represents epistemic knowledge \cite{Leifer-review,PBR,Leifer-Maroney,maroney2012statistical,Barrett,Leifer,Branciard,Spekkens-toy,Renner,hance2022wave}. While initially appearing metaphysical, the ontological models framework offers a comprehensive approach to tackle such questions \cite{harrigan2010einstein}. This framework posits the existence of underlying ontic states (representing reality) where the actual description of a pure quantum system is expressed through a probability distribution over these ontic states. If a single ontic state can correspond to multiple quantum preparations, the model is termed epistemic. Interestingly, a well-known no-go theorem for epistemic models was presented by Pusey-Barrett-Rudolph \cite{PBR}. However, this result relies on auxiliary assumptions, without which successful epistemic models are constructed \cite{Lewis,Aaronson}. Nonetheless, as demonstrated by Barrett \textit{et al.} \cite{Barrett}, Leifer \cite{Leifer}, and Branciard \cite{Branciard}, all epistemic models fail to explain the distinguishability of quantum states in an asymptotic way. 

Notably, in these prior works, the epistemic models are probed to explain the distinguishability of two quantum pure preparations. 
This research delves into a more fundamental level, exploring whether any epistemic model can accurately reproduce the empirical predictions of quantum mixed preparations.
For an arbitrary number of quantum preparations, the epistemic overlap is determined by the common overlap between the probability distribution over ontic states. As we show here, the operational counterpart of epistemic overlap, referred to as quantum overlap, is intricately linked to the average probability of anti-distinguishing a set of preparations. A set of quantum preparations is considered \textit{non-epistemic} if the epistemic overlap is zero for all possible ontological models while the quantum overlap remains non-zero. In its strongest form, a set is \textit{fully non-epistemic} if the epistemic overlap vanishes for all ontological models while the quantum overlap reaches its maximum value of one. A weaker form, termed \textit{non-maximally epistemic}, implies the epistemic overlap must be strictly less than the quantum overlap.

This work uncovers several fundamental yet unexplored aspects of quantum theory. Our key findings are as follows. First, we identify several sets of quantum mixed preparations in qubit systems that are non-epistemic concerning the common overlap between three preparations. Second, we present a class of quantum mixed preparations in qutrit and ququart systems that are fully non-epistemic with respect to the common overlap between four and three mixed preparations, respectively. Additionally, we show that fully non-epistemic mixed preparations cannot exist in qubit systems when considering overlaps between an arbitrary number of preparations. 

Our final set of results focuses on the overlap between two mixed states. Here, we establish two remarkable findings: First, we prove that mixed states composed of mutually unbiased bases become fully non-epistemic as the dimension of the Hilbert space increases. Second, we show that in any Hilbert space with a dimension greater than four, there exist mixed states for which the epistemic overlap asymptotically approaches zero relative to the quantum overlap. The former result represents the strongest manifestation of no-go theorems for epistemic models in the context of mixed preparations and signifies the most compelling form of preparation contextuality \cite{Spekkens2005}. In general, preparation contextuality implies that a given pair of mixed preparations is non-maximally epistemic; however, the reverse is not necessarily true. To explicitly demonstrate this implication, we obtain a relation between the quantum-to-epistemic overlap gap and the success metric of the parity-oblivious multiplexing task, which serves as a witness for preparation contextuality in the simplest possible scenario. Besides, we determine the exact maximum possible value of the difference between quantum and epistemic overlap for pairs of qubit mixed preparations. 

All these results are obtained by establishing upper bounds on the epistemic overlap using empirically observed quantities, particularly the anti-distinguishability. Consequently, these findings are robustly testable in the prepare-and-measure experiments. 
It is worth mentioning that not every non-epistemic proof for two preparations leads to the rejection of epistemic models for pure preparations, namely, $\psi-$epistemic models \cite{PBR,Leifer-review}. We identify the instances where such refutation occurs and note that certain examples presented in this study fall within this category. 

The paper is structured as follows. In the next section, we introduce the common epistemic and quantum overlap for an arbitrary number of general quantum preparations. The following section introduces different notions of non-epistemicity and their relationship to previously proposed notions of non-classicality. The subsequent section contains the main results, which are divided into two subsections: the first addresses the common overlap between three and four preparations, and the second concerns the overlap between two preparations. 
In the final section, we provide insights into our findings and outline directions for potential future work.

\section{Anti-distinguishability and common overlap between epistemic states}

As an \textit{operational theory}, quantum theory lays down the description of \textit{preparation procedures} denoted in general by density operators, \textit{transformation procedures} described by CPTP (completely positive trace-preserving) maps, and \textit{measurement procedures} described by POVM (positive operator-valued measurements) operators. Once all these procedures are specified, the operational theory provides a rule using which one is capable of predicting the likelihood of the measurement outcomes --- in this case, Born's rule. However, such a theory does not specify much about the system itself, since its first and foremost purpose is to predict experimental outcomes. An ontological model, on the other hand, seeks to explain the results of the operational theory in terms of attributes inherent to the system.

In the framework of ontological models, the complete specification of observer-independent attributes of a physical system, or the `real state of affairs' \cite{cs} of that system, is described by an \textit{ontic state} $\lambda$, which belongs to an \textit{ontic state space} $\Lambda$. Once a preparation procedure is done, the system occupies a specific $\lambda$. However, owing to the lack of fine-tuning of the preparation process and/or the inherent randomness present in nature, the quantum description $\rho$ of the system corresponds to many probable ontic states. Thus, $\rho$ is associated with an \textit{epistemic state} $\{\mu(\lambda|\rho)\}$, a distribution over the ontic state space.

Quantum mixed preparations are convex mixtures of pure states.  For a given mixed state $\varrho$, let one of its decompositions be $\{p_i;\psi_i\}_i\equiv \rho=\sum_ip_i|\psi_i\rangle\!\langle\psi_i|$, where $\{p_i\}_i$ are the convex coefficients and $\{\psi_i\}_i$ are the pure states. It is natural to assume that the corresponding epistemic state preserves this convexity, that is,
\be \label{fcm}
\forall \lambda, \ \ \mu(\lambda|\rho)=\sum_ip_i\ \mu(\lambda|\psi_i).
\ee 
Violation of this assumption would imply that the choice of preparing a mixed state is correlated with the ontic state $\lambda.$ Since a mixed state can be decomposed in infinitely many different ways, every mixed state has a corresponding set of epistemic states. Thus, here we make an important distinction between \textit{quantum mixed states} and \textit{quantum mixed preparations}. A quantum mixed state $\varrho$ is an operator that acts on the system's Hilbert space and can have many different convex decompositions. On the other hand, a quantum mixed preparation denotes a particular convex decomposition of the mixed state, and thus corresponds to a fixed set of convex coefficients and pure states that feature in that decomposition. Henceforth in this article, we shall follow the aforementioned notation, where a mixed state shall be denoted by $\varrho$, and any particular decomposition of it would be denoted by $\rho$ followed by either an explicit mention of the set of convex coefficients and pure states $\{p_i;\psi_i\}_i$ or an equation of the form $\rho=\sum_ip_i|\psi_i\rangle\!\langle\psi_i|$.

 Measurement in this framework is described by a set of conditional probabilities known as \textit{response functions}  $\{\xi(k|\mathcal{M},\lambda)\}_{k=1}^{n}$, which denotes the probability of obtaining an outcome $k$ given the system is in an ontic state $\lambda$, and the measurement procedure is $\mathcal{M}$. They obey the standard conditions imposed on conditional probabilities, which are $\xi(k|\mathcal{M},\lambda)\geq0$ and $\sum_k\xi(k|\mathcal{M},\lambda)=1$. Thus, the probability of obtaining an outcome $k$ after performing a measurement $\mathcal{M}=\{M_k\}_{k=1}^n$ on a quantum preparation $\rho$ is
\begin{equation}\label{prob}
    \tr(M_k\rho) =p\left(k|\mathcal{M},\rho\right)=\int_{\Lambda}\mu(\lambda|\rho)\xi(k|\mathcal{M},\lambda)d\lambda.
\end{equation}
In this article, we focus on a particular operational quantity known as the \textit{anti-distinguishability} of a set of quantum preparations, $\{\rho_x\}_x$. Anti-distinguishability is a prepare-and-measure game, where the player picks out any preparation from the set at random, and tries to guess which one they have \textit{not} picked out. If this task can be performed without error for all the preparations in the set, we call that set anti-distinguishable. Our goal is to explain anti-distinguishability using the framework of ontological models and identify sets of preparations where epistemic models fail to explain anti-distinguishability.

\begin{definition}[Anti-distinguishability]
    \textit{Anti-distinguishability} of a set of $n$ quantum preparations $\{\rho_x\}_{x=1}^n$ is defined as the average over the probabilities of outcomes of a $n-$outcome measurement procedure $\mathcal{M}=\{M_k\}_{k=1}^n$, such that the outcomes $k\neq x$ when the quantum preparation $\rho_x$ is being measured. The quantity is made measurement independent by choosing that $n-$outcome measurement procedure for which this average is maximum,
\begin{equation}\label{anti-d}
    \begin{split}
    A_Q^{[n]}=&\frac{1}{n}\max_{\{\mathcal{M}\}}\left(\sum_{x,k=1}^np\left(k\neq x|\rho_x,\mathcal{M}\right)\right)\\
    =&1-\frac{1}{n}\min_{\{\mathcal{M}\}}\left(\sum_{x=1}^np(x|\rho_x,\mathcal{M})\right)\\
    = & 1-\frac{1}{n}\min_{\{\mathcal{M}\}}\left(\sum_{x=1}^n\tr(M_x\rho_x)\right).
    \end{split}
\end{equation}
\end{definition}

Anti-distinguishability of a set of preparations can be connected to their epistemic states by expressing the probabilities of the outcomes by \eqref{prob}. As it turns out, this leads to a quantity known as the \textit{common epistemic overlap}. Note that we have used the superscript $[n]$ to denote the number of preparations under consideration.

\begin{definition}[Common epistemic overlap of $n$ preparations]
     The \textit{common epistemic overlap} of $n$ quantum preparations $\{\rho_x\}_{x=1}^n$ is defined as the common overlap between the epistemic states corresponding to the quantum preparations,
     \begin{equation}\label{epis_overlap}
         \omega_E^{[n]}\left(\rho_1,\cdots,\rho_n\right)=\int_{\Lambda}\min\left(\mu(\lambda|\rho_1),\cdots,\mu(\lambda|\rho_n)\right)d\lambda.
     \end{equation}
\end{definition}
The connection between $\omega_E^{[n]}$ and $A_Q^{[n]}$ is made by replacing the probabilities in \eqref{anti-d} with those defined in \eqref{prob},
\begin{equation} \label{we-aq}
    \begin{split}
        \sum_{x=1}^np(x|\rho_x,\mathcal{M})&=\sum_{x=1}^n\left(\int_{\Lambda}\mu(\lambda|\rho_x)\xi(x|\lambda,\mathcal{M})d\lambda\right)\\
        &=\int_{\Lambda}\left(\sum_{x=1}^n\mu(\lambda|\rho_x)\xi(x|\lambda,\mathcal{M})\right)d\lambda\\
        &\geq\int_{\Lambda}\min(\mu(\lambda|\rho_1),\cdots,\mu(\lambda|\rho_n))d\lambda\\
    &=\omega_E^{[n]}(\rho_1,\cdots,\rho_n).
    \end{split}
\end{equation}
The third line in \eqref{we-aq} follows from the fact that $\sum_{x=1}^n\xi(x|\lambda,\mathcal{M})=1$ for any given $\mathcal{M}$ and for all $\lambda \in \Lambda$. From \eqref{we-aq}, it becomes evident that
\begin{equation}\label{AQ_epis}
  \omega_E^{[n]}\leq n\left(1-A_Q^{[n]}\right).
\end{equation}
That there exists a relationship between $\omega_E^{[n]}$ and $A_Q^{[n]}$ is evident from \eqref{AQ_epis}. To better capture its essence, we introduce a new operational quantity known as the \textit{common quantum overlap}, which contains the same amount of information as $A_Q^{[n]}$, but whose interpretation is subtly different and better suited to establish a connection with the underlying ontological model.

\begin{definition}[Common quantum overlap of $n$ preparations]
    \textit{Common quantum overlap} of $n$ quantum preparations is defined in accordance with (\ref{AQ_epis}), and it encapsulates the degree to which the preparations are not anti-distinguishable,
    \begin{equation}\label{Q_overlap}
        \omega_Q^{[n]}=n\left(1-A_Q^{[n]}\right).
    \end{equation}
\end{definition}
 For $n$ quantum preparations, at most $A_Q^{[n]}=1$ and at least $1-\frac{1}{n}$, which implies $\omega_Q^{[n]}\in[0,1]$. From \eqref{AQ_epis} and \eqref{Q_overlap},
 \begin{equation}\label{we<wq}
     \omega_E^{[n]}(\rho_1,\cdots,\rho_n)\leq\omega_Q^{[n]}(\rho_1,\cdots,\rho_n).
 \end{equation}
 In the case of an anti-distinguishable set of quantum preparations, $\omega_E^{[n]}=\omega_Q^{[n]}=0$. 
The above relation \eqref{we<wq} represents the degree to which an epistemic theory can explain the prediction of the quantum theory. Further elaboration on this has been provided in the next section. Here, we offer a brief discussion on \textit{distinguishability} and how it is relevant for our purposes.
 
It is not difficult to see that for $n=2$, anti-distinguishability \eqref{anti-d} is equal to minimum-error-state discrimination \cite{SchmidPRX}, also referred to as distinguishability \cite{cs}. Let us consider two quantum preparations, $\rho_1$ and $\rho_2$. Distinguishability is then an optimization over every two-outcome measurement $\mathcal{M}=\{M_1,M_2\}$ such that
\begin{equation} \label{dist_def}
\begin{split}
    D_Q(\rho_1,\rho_2)=&\frac{1}{2}\max_{\{\mathcal{M}\}}\left(\tr(M_1\rho_1)+\tr(M_2\rho_2)\right)\\
    =&\frac{1}{2}\max_{\{\mathcal{M}\}}\left(\tr(\tilde{M}_2\rho_1)+\tr(\tilde{M}_1\rho_2)\right)\\
    =& A_Q(\rho_1,\rho_2),
    \end{split}
\end{equation} 
where we have relabelled $M_1\rightarrow \tilde{M}_2$ and $M_2\rightarrow \tilde{M}_1$.
Thus for $n=2$,
\be \label{Lqn2}
\omega_Q = 2 (1 - D_Q).
\ee
Eq. \eqref{dist_def} has a closed-form expression,
\begin{equation} \label{distin}
    D_Q(\rho_1,\rho_2)=\frac{1}{2}\left(1+T(\rho_1,\rho_2)\right),
\end{equation}
where $T(\rho_1,\rho_2)$ is known as the \textit{trace distance} between $\rho_1$ and $\rho_2$. It serves as a metric on the space of density matrices and is defined as,
\begin{equation}
\begin{split}
    T(\rho_1,\rho_2) =\frac{1}{2}\tr(|\rho_1-\rho_2|) = \frac{1}{2}\sum_{i=1}^r|\alpha_i|,
    \end{split}
\end{equation}
 where $\alpha_i$ is the $i-$th eigenvalue of $\rho_1-\rho_2$ and $r$ is its rank. For pure states $\ket{\psi}$ and $\ket{\phi}$ it turns out that $T(\psi,\phi)=\sqrt{1-|\bra{\psi}\phi\rangle|^2}$ and consequently,
\begin{equation} \label{distin_pure}
    D_Q(\psi,\phi)=\frac{1}{2}\left(1+\sqrt{1-|\bra{\psi}\phi\rangle|^2}\right).
\end{equation}
These definitions will become important in Section \ref{sec3_2}, where we have looked into the epistemic overlap of two mixed preparations.

Note that we have omitted the superscript $[n]$ while dealing with cases where $n=2$, which is the convention we have followed throughout this article.

\section{Notions of non-epistemicity for mixed preparations} \label{sec3-1}

Quantum overlap $\omega_Q^{[n]}$ of a set of preparations is a characteristic of that set, and quantum theory, being an operational theory, provides a prescription to predict its value. However, to seek a realist explanation for why this value takes a particular form, one must turn to ontological models.
An epistemic explanation for quantum overlap suggests that some ontic states appear in the epistemic states corresponding to multiple preparations, leading to non-zero epistemic overlap $\omega_E^{[n]}$ which in turn leads to non-zero $\omega_Q^{[n]}$. If the epistemic overlap fully accounts for the anti-distinguishability of a set of preparations, meaning $\omega_E^{[n]}$ precisely matches $\omega_Q^{[n]}$, then we have a maximally epistemic explanation.
\begin{definition}[Maximally epistemic]
  A set of $n$ preparations $\{\rho_x\}_{x=1}^n$ is maximally epistemic if there exists at least one ontological model in which the epistemic overlap $\omega_E^{[n]}$ is equal to the quantum overlap $\omega_Q^{[n]}$, 
        $\omega_E^{[n]} = \omega_Q^{[n]}.$
\end{definition}
If no ontological model satisfies $\omega_E^{[n]} = \omega_Q^{[n]}$ and instead $\omega_E^{[n]}<\omega_Q^{[n]}$ for all possible ontological models, the set is non-maximally epistemic. 
%
\begin{definition}[Non-maximally epistemic]
  A set of $n$ mixed preparations $\{\rho_x\}_{x=1}^n$ is non-maximally epistemic if the epistemic overlap $\omega_E^{[n]}$ is strictly less than the quantum overlap $\omega_Q^{[n]}$ for all possible ontological models, $\omega_E^{[n]}<\omega_Q^{[n]}$.
\end{definition}
Nonetheless, in this case, anti-distinguishability has an epistemic explanation but is limited by coarse-grained measurements that fail to identify individual ontological states. This fact is well encapsulated by Eq.~\eqref{we<wq}. However, if the epistemic overlap $\omega_E^{[n]}$ is zero compared to the quantum overlap $\omega_Q^{[n]}$ for all possible ontological models, then no ontic state appears in all the preparations.  This signifies the absence of any epistemic explanation for the quantum overlap, leading to the following definition.
\begin{definition}[Non-epistemic]
 A set of $n$ mixed preparations $\{\rho_x\}_{x=1}^n$ is non-epistemic if $\omega_E^{[n]}=0$ for all possible ontological models while $\omega_Q^{[n]}>0$.
\end{definition}
The most extreme case of non-epistemicity occurs when the quantum overlap reaches its maximum possible value, meaning that the mixed preparations yield identical quantum states. Such a scenario can arise exclusively when dealing with mixed preparations. We define this case as follows.
\begin{definition}[Fully non-epistemic]
    A set of $n$ mixed preparations $\{\rho_x\}_{x=1}^n$ is fully non-epistemic if the epistemic overlap $\omega_E^{[n]}=0$ while the quantum overlap $\omega_Q^{[n]}=1$ for all possible ontological models.
\end{definition}
In this work, non-maximally epistemic cases are briefly addressed in Section \ref{sec3_2}. However, the primary focus is on identifying and analyzing non-epistemic and fully non-epistemic cases.

It is worth noting that the concept of a maximally epistemic model for pure preparations has been previously proposed, and its limitations have been studied in \cite{Leifer-Maroney,Leifer-review,Leifer,Barrett,Branciard}. For two mixed preparations, non-maximally epistemicity is referred to as `excess ontological distinctness' in \cite{cs}, and its existence has been shown. On the other hand, preparation contextuality, introduced by Spekkens \cite{Spekkens2005}, is a specific instance of non-maximally epistemic mixed preparations. Preparation non-contextuality, the negation of this concept, asserts that any two mixed preparations yielding identical quantum mixed states must correspond to the same epistemic distributions. Note that for two identical mixed preparations $\rho_1$ and $\rho_2$, the distinguishability is 1/2, which is equivalent to $\omega_Q(\rho_1,\rho_2)=1$ due to \eqref{Lqn2}. Under preparation non-contextuality, one therefore has $\omega_E(\rho_1,\rho_2)=1$. Conversely, any violation of preparation non-contextuality implies $\omega_E(\rho_1,\rho_2) < \omega_Q(\rho_1,\rho_2)=1.$ Preparation contextuality has been established as a cornerstone of the nonclassical features of quantum theory \cite{Spekkens2005,Review,PuseyPRA2018,PuseyPRL,Spekkens2008prl,interference,UR,mate-prl-2023,Chaturvedi2021characterising,Plavala-prl,Shin2021QuantumContextual,XuPRA} and plays a crucial role in quantum information processing and communication \cite{Spekkens2009,SikoraNJP,SahaNJP,sahaPRA,PanPRA,SchmidPRL22,SchmidPRX,LostaglioPRL,Randomness_certify,State_discri_2}.

\begin{definition}[Preparation contextual]
    A pair of mixed preparations $\{\rho_1,\rho_2\}$ is said to be preparation contextual, if $\omega_Q(\rho_1,\rho_2)=1$ but $\omega_E(\rho_1,\rho_2)<1$.
\end{definition}

 As discussed in Section \ref{sec3_2}, there exist pairs of non-maximally epistemic mixed states that nevertheless do not exhibit preparation contextuality. In particular, we demonstrate the existence of non-identical qubit states $\rho_1$ and $\rho_2$ such that $\omega_E(\rho_1,\rho_2)<\omega_Q(\rho_1,\rho_2)<1$.  

The strongest manifestation of preparation contextuality arises when the epistemic overlap $\omega_E(\rho_1,\rho_2)$ vanishes, which is equivalent to fully non-epistemicity for two mixed preparations. We introduce this formally as follows.
\begin{definition}[Fully preparation contextual or fully non-epistemic pairs of preparations]
    A pair of mixed preparations $\{\rho_1,\rho_2\}$ is said to be fully preparation contextual if $\omega_E(\rho_1,\rho_2) = 0$ while $\omega_Q(\rho_1,\rho_2) = 1$.
\end{definition}
We will show later that fully preparation contextual mixed preparations exist as the dimension of the states increases.

\section{Non-epistemic mixed preparations} \label{sec3-2}

Following, we present the theorem that forms the backbone of this article, from which many of the results follow.

\begin{widetext}
 \begin{thm} \label{t1}
Consider a set of $n$ mixed preparations $\{\rho_k\}_{k=1}^n$, each of which is a convex mixture of $m$ pure states $\{\ket{\psi_{i_k|k}}\}_{i_k=1}^m$ with arbitrary convex coefficients $\{p_{i_k|k}\}_{i_k=1}^m$, where subscript $k$ denotes the mixed preparation and the subscript $i_k|k$ denotes the pure state $\ket{\psi_{i_k|k}}$ belonging to the decomposition of $\rho_k$. Thus, $\rho_k=\sum_{i_k=1}^{m}p_{i_k|k}|\psi_{i_k|k}\rangle\!\langle\psi_{i_k|k}|$. Then the following relation holds for this set of mixed preparations,
    \begin{subequations} \label{wr<ws}
    \begin{align}
        \omega_E^{[n]}\left(\rho_1,\rho_2,\cdots,\rho_n\right) &\leq\sum_{i_1,i_2,\cdots,i_n=1}^{m}\max\left(p_{i_1|1},p_{i_2|2}\cdots,p_{i_n|n}\right)\omega_{E}^{[n]}\left(\psi_{i_1|1},\psi_{i_2|2},\cdots,\psi_{i_n|n}\right) \label{sub_eq1}\\
        &\leq n\sum_{i_1,i_2,\cdots,i_n=1}^{m}\max\left(p_{i_1|1},p_{i_2|2}\cdots,p_{i_n|n}\right)\Big(1-A_{Q}^{[n]}\left(\psi_{i_1|1},\psi_{i_2|2},\cdots,\psi_{i_n|n}\right)\Big). \label{sub_eq2}
        \end{align}
    \end{subequations}
 Moreover, if the left-hand-side of \eqref{sub_eq1} is zero, then the inequality becomes equality, that is, each term on the right-hand-side is zero.
\end{thm}
\end{widetext}
The proof of this theorem is provided in Appendix \ref{app:pot1}. This theorem allows us to establish a relation between the epistemic overlap of $n$ mixed preparations and the epistemic overlap of their decompositions. Notice that the theorem also applies to sets of mixed preparations with different numbers of pure states in their decompositions, even though in the theorem statement, the number is the same throughout, which is $m$. This is possible because we can assume $m$ to be the maximum number of pure states among all the decompositions of the mixed preparations, and for those preparations that contain pure states less than $m$, we can do the following. Suppose every other preparation is made out of $m=3$ pure states and one is made out of two, $\rho=p\ket{\psi_1}\!\bra{\psi_1}+(1-p)\ket{\psi_2}\!\bra{\psi_2}$. One can rewrite $\rho$ as $\rho=p/2\ket{\psi_1}\!\bra{\psi_1}+p/2\ket{\psi_1}\!\bra{\psi_1}+(1-p)\ket{\psi_2}\!\bra{\psi_2}$. The first two pure states in the decomposition are the same with modified convex coefficients, making $m=3$ for this preparation as well. Such a trick has been used while applying Theorem \ref{t1} to Example \ref{ex_1} in Section \ref{sec_A_1}.

A given mixed state corresponds to infinitely many different preparations, and this inequality holds for all such preparations. That leads to some interesting features, such as the scenario where we have $n$ maximally mixed preparations, each one being decomposed into a complete orthonormal set of pure states distinct from the other, leading to consequences relevant to preparation contextuality (discussed in detail in Section \ref{sec5}), or cases where the decompositions are perfectly anti-distinguishable, making the mixed preparations' epistemic overlap vanish. 

 It is important to note that since the epistemic overlap of the mixed preparations is bounded by the sum of the quantum overlap of its decompositions \eqref{sub_eq2}, which are operational quantities, it is an experimentally robust way to obtain an upper bound on the epistemic overlap.

Below, we have provided a corollary of Theorem \ref{t1}, which deals with epistemic and quantum overlaps of $n$ maximally mixed preparations. 

 \begin{coro} \label{c1}
    Consider a set of $n$ maximally mixed preparations $\{\rho_k\}_{k=1}^n$, all of which act on the Hilbert space $\mathbbm{C}^d$. Each one can be decomposed into a complete orthonormal set of pure states $\{\ket{\psi_{i_k|k}}\}_{i_k=1}^d$ such that $\rho_k=(1/d)\sum_{i_k=1}^{d}|\psi_{i_k|k}\rangle\!\langle\psi_{i_k|k}|=\mathbbm{1}/d$ for all $k$. Evidently, quantum overlap $\omega_Q^{[n]}(\rho_1,\cdots,\rho_n)=1$. Then the following relation holds:
        \bea \label{coro1_eq}
                &  & \omega_E^{[n]}(\rho_1,\cdots,\rho_n) \nonumber \\
                &\leq & \frac{1}{d}\sum_{i_1,\cdots,i_n=1}^{d}\omega_{E}^{[n]}\left(\psi_{i_1|1},\cdots,\psi_{i_n|n}\right)\nonumber \\
                &\leq & \frac{n} {d}\sum_{i_1,\cdots,i_n=1}^{d}\left(1-A_Q^{[n]}\left(\psi_{i_1|1},\cdots,\psi_{i_n|n}\right)\right) .
            \eea 
\end{coro}
  \begin{proof}
    The proof follows from Theorem \ref{t1}. Instead of having $n$ general sets of pure states, we have $n$ complete and orthonormal sets each containing $d$ number of states. The convex coefficients are $p_{i_k|k}=1/d$ for all $i_k \in \{1,\cdots,d\}$ and $k \in \{1,\cdots,n\}$. Substituting these values in \eqref{wr<ws}, we obtain \eqref{coro1_eq}. 
\end{proof} 
    
\subsection{Sets of more than two non-epistemic mixed preparations} \label{sec3_1}

In the following sub-sub-sections, we will use Theorem \ref{t1} and Corollary \ref{c1} to present criteria that can be used to identify non-epistemic and fully non-epistemic quantum preparations.

\subsubsection{Non-epistemic preparations in $d=2$}\label{sec_A_1}
First, we provide an explicit example of a set of non-epistemic preparations. 
\begin{example}[Non-epistemic, $n=3,d=2$] \label{ex_1}
    Consider the following set of qubit preparations, $\rho_1=|0\rangle\!\langle0|,\rho_2=|+\rangle\!\langle+|,\rho_3=\frac{1}{2}\left(|1\rangle\!\langle1|+|-\rangle\!\langle-|\right)$. One can show that, $\omega_E^{[3]}\left(\rho_1,\rho_2,\rho_3\right)=0$ while $\omega_Q^{[3]}\left(\rho_1,\rho_2,\rho_3\right)=0.1161$.
\end{example}

    Invoking Theorem \ref{t1} with respect to the given preparations we see that $p_{1|1}=p_{2|1}=p_{1|2}=p_{2|2}=p_{1|3}=p_{2|3}=1/2$. Substituting these in \eqref{wr<ws} we have, 
    \begin{equation}\label{eqex}
    \begin{split}
\omega_E^{[3]}\left(\rho_1,\rho_2,\rho_3\right)&\leq 6\Big(2-A_Q^{[3]}(|0\rangle,|+\rangle,|1\rangle)\\
& \qquad -A_Q^{[3]}(|0\rangle,|+\rangle,|-\rangle)\Big).
    \end{split}
    \end{equation}
    Consider the set $\{|0\rangle,|+\rangle,|1\rangle\}$. These three states are perfectly anti-distinguishable because there exists a POVM set $\mathcal{M}\equiv\{M_1=|1\rangle\!\langle1|, M_2=\mathbbm{O}, M_3=|0\rangle\!\langle0|\}$ such that $p(1|\mathcal{M},|0\rangle)+p(2|\mathcal{M},|+\rangle)+p(3|\mathcal{M},|1\rangle)=0$. Then according to \eqref{anti-d}, $A_Q^{[3]}\left(|0\rangle,|+\rangle,|1\rangle\right)=1$. A similar argument can be made for the set $\{|0\rangle,|+\rangle,|-\rangle\}$ and hence $A_Q^{[3]}(|0\rangle,|+\rangle,|-\rangle)=1$ as well. So according to \eqref{eqex},  $\omega_E^{[3]}(\rho_1,\rho_2,\rho_3)=0$.
    However, a quick semi-definite optimization yields $A_Q^{[3]}(\rho_1,\rho_2,\rho_3)=0.9613,$ which makes $\omega_Q^{[3]}(\rho_1,\rho_2,\rho_3)=0.1161>0$. Thus, this is an example of a set of non-epistemic preparations.

Now that we have provided an explicit example of a non-epistemic case, we aim to provide a set of criteria that can be used to identify a large class of such examples. To begin, we derive the necessary and sufficient conditions for three qubit states to be perfectly anti-distinguishable. These conditions offer a clear interpretation concerning the orientation of their Bloch vectors within the Bloch sphere. It must be noted that such a set of sufficient conditions exists, as shown in \cite{caves}. In \cite{caves}, the criteria for anti-distinguishability among three qutrit states under projective measurements have been established. Consequently, these criteria are sufficient. However, for qubits, the criteria are also necessary, as any rank-one POVM on qubits can be realized through projective measurements on three-dimensional quantum states.

\begin{thm} \label{bloch_anti_d_cond}
    Three qubit states $\ket{\psi_1}, \ket{\psi_2},$ and $\ket{\psi_3}$ are perfectly anti-distinguishable if and only if all of them lie on a single great circle of the Bloch sphere and the following inequalities hold,
    \begin{equation} \label{bloch_cond}
    \begin{split}
    &\cos^{-1}{(a)} + \cos^{-1}{(b)} \geq  \frac{\pi}{2},\\
    &\cos^{-1}{(a)} + \cos^{-1}{(c)} \geq  \frac{\pi}{2},\\
    &\cos^{-1}{(b)} + \cos^{-1}{(c)} \geq  \frac{\pi}{2}, 
    \end{split}
    \end{equation} 
     where $a=|\la \psi_1|\psi_2\ra|,$ $b=|\la \psi_1|\psi_3\ra|,$ $c=|\la \psi_2|\psi_3\ra|.$ In other words, only those three qubit states are perfectly anti-distinguishable whose Bloch vectors lie on a single great circle, and the sum of every pair of angles between the vectors is greater than or equal to $\pi$.
\end{thm}

See Appendix \ref{app:pot2} for the proof of this theorem. With the aid of the above theorem, we present a class of non-epistemic examples involving only four pure qubit preparations.

  \begin{thm} \label{cond_non_epis_thm}
    Three qubit preparations $\rho_1 = |\psi\ra\!\la\psi|,$ $\rho_2 = |\phi\ra\!\la\phi|,$ and $\rho_3 = p|\chi_1\ra\!\la\chi_1| + (1-p)|\chi_2\ra\!\la\chi_2| ,$ for $0<p<1$, are non-epistemic if $\ket{\psi}$, $\ket{\phi}$, $\ket{\chi_1}$ and $\ket{\chi_2}$ lie on a single great circle of the Bloch sphere, $a=|\langle\psi|\phi\rangle|>0$ and the following inequalities hold,
\begin{equation} \label{cond_1_non_epis}
\begin{split}
    &\cos^{-1}(a)+\cos^{-1}(b)\geq\pi/2,\\
    &\cos^{-1}(a)+\cos^{-1}(c)\geq\pi/2,\\
    &\cos^{-1}(b)+\cos^{-1}(c)\geq\pi/2,\\
    &\cos^{-1}(a)+\cos^{-1}(d)\geq\pi/2,\\
    &\cos^{-1}(a)+\cos^{-1}(e)\geq\pi/2,\\
    &\cos^{-1}(d)+\cos^{-1}(e)\geq\pi/2,
\end{split}
\end{equation}
where $b=|\langle\psi|\chi_1\rangle|$, $c=|\langle\phi|\chi_1\rangle|$, $d=|\langle\psi|\chi_2\rangle|$, and $e=|\langle\phi|\chi_2\rangle|$.
\end{thm}
\begin{proof}
    By invoking Theorem \ref{t1} concerning the preparations described in the theorem statement, and by substituting the convex coefficients in \eqref{wr<ws} we have (it has been assumed without loss of generality that $p>1/2$),
        \bea 
            \omega_E^{[3]}(\rho_1,\rho_2,\rho_3) & \leq & \Big((18p+3)-12pA_Q^{[3]}(\psi,\phi,\chi_1) \nonumber \\
            & & \quad -(6p+3)A_Q^{[3]}(\psi,\phi,\chi_2)\Big).
            \eea
        To have a non-epistemic case, $\{\ket{\psi},\ket{\phi},\ket{\chi_1}\}$ and $\{\ket{\psi},\ket{\phi},\ket{\chi_2}\}$ must be perfectly anti-distinguishable sets. According to Theorem \ref{bloch_anti_d_cond}, the states in the set $\{\ket{\psi},\ket{\phi},\ket{\chi_1}\}$ must lie on a great circle of the Bloch sphere and so do the states in $\{\ket{\psi},\ket{\phi},\ket{\chi_2}\}$. Clearly, then, all four states must lie on the same great circle. Next, they must satisfy the inequalities laid down in \eqref{bloch_cond}. It is quite easy to see that \eqref{cond_1_non_epis} follows from these inequalities. We have also stipulated that $a=|\langle\psi|\phi\rangle|>0$, and that is because we want a non-zero quantum overlap between $\rho_1$, $\rho_2$, and $\rho_3$. Since $\rho_1$ and $\rho_2$ are the pure states $\ket{\psi}$ and $\ket{\phi}$, it is imperative that they are not orthogonal to have $\omega_Q^{[3]}(\rho_1,\rho_2,\rho_3)>0$. 
    \end{proof} 
    
    One must note that Example \ref{ex_1} belongs to the class of preparations described in Theorem \ref{cond_non_epis_thm}. Other such examples can be identified by applying the conditions described therein.

Now that we have shown the existence of non-epistemic cases, naturally, one must wonder about the existence of fully non-epistemic ones. The next theorem negates their existence in Hilbert space of $d=2$. 
\begin{thm}\label{thm_no_fully_epis} There are no sets of mixed qubit preparations that are fully non-epistemic. In other words, there exists a model for qubits such that maximal quantum overlap necessarily implies a nonzero epistemic overlap, i.e., $\omega^{[n]}_Q(\rho_1, \cdots,\rho_n)=1 \implies \omega^{[n]}_E(\rho_1, \cdots,\rho_n) > 0.$
\end{thm}
The proof of this theorem is extensive, and for the sake of clarity in presentation, it has been deferred to Appendix \ref{app:pot3}.
In essence, the above theorem holds due to the existence of an epistemic theory for pure qubit states, known as the Renner-Tavakoli-Quintino (RTQ) model \cite{RTQ}. This model is capable of simulating quantum statistics arising out of general measurements performed on a qubit. Let us provide a brief introduction to this model, as we will refer to this model later on in Appendix \ref{app:pot3}. The ontic state in the RTQ model is a collection of two classical bits and two three-dimensional unit vectors, $\lambda=(c_1,c_2;\vec{r}_1,\vec{r}_2)$, where $c_1,c_2 \in \{0,1\}$ and $\vec{r}_1,\vec{r}_2 \in \mathcal{S}^2$ which implies that $\Lambda=\{0,1\}\otimes\{0,1\}\otimes \mathcal{S}^2\otimes \mathcal{S}^2$ \blk. Let the Bloch representation of a pure qubit preparation be $\vec{v}$, where $|\vec{v}|=1$. Then the corresponding epistemic state is defined as 
\be
\mu(c_1,c_2;\vec{r}_1,\vec{r}_2|\vec{v})=\frac{1}{16\pi^2}\left(\delta_{c_1,\Theta(\vec{v}.\vec{r}_1)}\delta_{c_2,\Theta(\vec{v}.\vec{r}_2)}\right),
    \ee
where $\Theta(.)$ is the Heaviside step function. Thus, $\mu(c_1,c_2;\vec{r}_1,\vec{r}_2|\vec{v})>0$ when $\Theta(\vec{v}.\vec{r}_1)=c_1$ and $\Theta(\vec{v}.\vec{r}_2)=c_2$. For $c_1\neq c_2$, the support of $\mu(c_1,c_2;\vec{r}_1,\vec{r}_2|\vec{v})$ is the entire unit sphere with $\vec{r}_1$ belonging to one hemisphere and $\vec{r}_2$ belonging to the opposite hemisphere. For $c_1=c_2=1$, both $\vec{r}_1$ and $\vec{r}_2$ belong to the hemisphere whose axis is $\vec{v}$ whereas if $c_1=c_2=0$ then $\vec{r}_1$ and $\vec{r}_2$ belong to the hemisphere that is opposite to the one whose axis is $\vec{v}$.

\subsubsection{Fully non-epistemic preparations in $d\geq3$}



Having shown the existence of non-epistemic cases and the non-existence of fully non-epistemic ones in Hilbert space of $d=2$, we extend our search for fully non-epistemic cases to higher dimensional Hilbert spaces. It turns out, they do exist, as is evident from the following two corollaries of Theorem \ref{t1} presented below.
\begin{coro} \label{n=4,d=3}
    Consider a set of four different preparations $\{\rho_k\}_{k=1}^4$ of the same mixed state $\varrho$ which acts on $\mathbbm{C}^3$. Each is a convex mixture of pure states $\{\ket{\psi_{i_k|k}}\}_{i_k=1}^m$ with convex coefficients $\{p_{i_k|k}\}_{i_k=1}^m$ such that $\rho_k=\sum_{i_k=1}^mp_{i_k|k}\ket{\psi_{i_k|k}}\!\langle\psi_{i_k|k}|$. Evidently, $\omega_Q^{[4]}(\rho_1,\rho_2,\rho_3,\rho_4)=1$. If the pure states obey the condition,
    \begin{equation} \label{sikora_eq}
        |\bra{\psi_{i_k|k}}\psi_{j_{k'}|k'}\rangle|\leq\frac{1}{\sqrt{3}} \textrm{ }\ \forall k\neq k', i_k, j_{k'}
    \end{equation}
    then $\omega_E^{[4]}(\rho_1,\rho_2,\rho_3,\rho_4)=0$. In other words, if the pure states satisfy \eqref{sikora_eq} then the set of mixed preparations is fully non-epistemic.
\end{coro}
\begin{proof}
   The epistemic overlap of the mixed preparations, $\omega_E^{[4]}(\rho_1,\rho_2,\rho_3,\rho_4)$ is upper-bounded according to \eqref{sub_eq2} of Theorem \ref{t1} with $n=4$. One way of making sure that $\omega_E^{[4]}(\rho_1,\rho_2,\rho_3,\rho_4)=0$ is to ensure that every such set $\{\ket{\psi_{i_1|1}},\ket{\psi_{i_2|2}},\ket{\psi_{i_3|3}},\ket{\psi_{i_4|4}}\}$, for all values of $i_1$, $i_2$, $i_3$, and $i_4$ is perfectly anti-distinguishable. We employ the criterion for anti-distinguishability as derived in \cite{johnston2023tight} to ensure this. According to \cite{johnston2023tight}, a set of $N$ pure states $\{\ket{\phi_i}\}_{i=1}^N$ is perfectly anti-distinguishable if,
   \begin{equation} \label{sikora_criterion}
       |\la\phi_i|\phi_j\ra|\leq \frac{1}{\sqrt{2}}\sqrt{\frac{N-2}{N-1}}  \textrm{ }\ \forall i \neq j.
   \end{equation}
   Since we need  $\{\ket{\psi_{i_1|1}},\ket{\psi_{i_2|2}},\ket{\psi_{i_3|3}},\ket{\psi_{i_4|4}}\}$ to be perfectly anti-distinguishable, and every such set has four pure states, we substitute $N=4$ in \eqref{sikora_criterion}. For $N=4$, the upper-bound in \eqref{sikora_criterion} is $1/\sqrt{3}$, and thus, the set of mixed preparations is fully non-epistemic when $|\bra{\psi_{i_k|k}}\psi_{j_{k'}|k'}\rangle|\leq 1/\sqrt{3}$.
\end{proof}

Based on Corollary \ref{n=4,d=3}, we present a very simple example of a set of mixed preparations which are fully non-epistemic.
\begin{example}[Fully non-epistemic, $n=4,d=3$] \label{ex_2}
 Consider the four sets of mutually unbiased bases (MUBs) present in $\mathbbm{C}^3$, $\left\{\left\{|\psi_{i_k|k}\rangle\right\}_{i_k=1}^{3}\right\}_{k=1}^{4}$. Each of these sets will result in a maximally mixed preparation, $\rho_k=(1/3)\sum_{i_k=1}^3|\psi_{i_k|k}\rangle\!\langle\psi_{i_k|k}|=\mathbbm{1}/3$, for $k=1,2,3,4$ such that $\omega_E^{[4]}(\rho_1,\rho_2,\rho_3,\rho_4)=0$ while $\omega_Q^{[4]}(\rho_1,\rho_2,\rho_3,\rho_4)=1$.
 \end{example}
 Since the states are mutually unbiased to each other, and $d=3$, so $|\bra{\psi_{i_k|k}}\psi_{j_{k'}|k'}\rangle|=1/\sqrt{3}$ for all $k \neq k'$, $i_k$ and $j_{k'}$. This is in accordance with \eqref{sikora_criterion}, and hence $\omega_E^{[4]}(\rho_1,\rho_2,\rho_3,\rho_4)=0$. 

\begin{coro} \label{n=3,d=4}
     Consider a set of three different preparations $\{\rho_k\}_{k=1}^3$ of the same mixed state $\varrho$ which acts on $\mathbbm{C}^d$ where $d\geq4$. Each one is a convex mixture of pure states $\{\ket{\psi_{i_k|k}}\}_{i_k=1}^m$ with convex coefficients $\{p_{i_k|k}\}_{i_k=1}^m$ such that $\rho_k=\sum_{i_k=1}^m p_{i_k|k}\ket{\psi_{i_k|k}}\!\langle\psi_{i_k|k}|$. Evidently, $\omega_Q^{[3]}(\rho_1,\rho_2,\rho_3)=1$. For every triplet of the form, $\{\ket{\psi_{i_1|1}},\ket{\psi_{i_2|2}},\ket{\psi_{i_3|3}}\}$ let $x_1=|\la\psi_{i_1|1}|\psi_{i_2|2}\ra|^2$, $x_2=|\la\psi_{i_1|1}|\psi_{i_3|3}\ra|^2$, and $x_3=|\la\psi_{i_2|2}|\psi_{i_3|3}\ra|^2$. We have $m^3$ such tuples $(x_1,x_2,x_3)$. For each such tuple, if the following conditions hold,
     \begin{equation} \label{caves_eq}
         \begin{split}
             &x_1+x_2+x_3<1,\\
             &(x_1+x_2+x_3-1)^2\geq 4x_1x_2x_3,
         \end{split}
     \end{equation}
     then $\omega_E^{[3]}(\rho_1,\rho_2,\rho_3)=0$. In other words, if the pure states satisfy \eqref{caves_eq}, then the set of mixed preparations is fully non-epistemic.
\end{coro}
\begin{proof}
    This proof is similar to that of Corollary \ref{n=4,d=3}. The upper-bound of $\omega_E^{[3]}(\rho_1,\rho_2,\rho_3)$, as described in \eqref{sub_eq2} for $n=3$ can be made zero by making sure that every such set $\{\ket{\psi_{i_1|1}},\ket{\psi_{i_2|2}},\ket{\psi_{i_3|3}}\}$, for all values of $i_1$, $i_2$, and $i_3$, is perfectly anti-distinguishable. To ensure this, we utilize the criteria provided in \cite{caves}. According to \cite{caves}, a set of three pure, non-orthogonal, and distinct states $\{\ket{\phi_1},\ket{\phi_2},\ket{\phi_3}\}$ is \textit{PP-incompatible}, meaning a projective measurement exists that is capable of anti-distinguishing them perfectly if, $x_1=|\langle \phi_1|\phi_2\rangle|^2$, $x_2=|\langle \phi_2|\phi_3\rangle|^2$ and $x_3=|\langle \phi_1|\phi_3\rangle|^2$ satisfy \eqref{caves_eq}. Applying \eqref{caves_eq} to $\{\ket{\psi_{i_1|1}},\ket{\psi_{i_2|2}},\ket{\psi_{i_3|3}}\}$ completes the proof.
\end{proof}
 Similar to Example \ref{ex_2}, we present the following example that demonstrates the fully non-epistemic nature of certain sets of preparations in Hilbert spaces of $d\geq4$.
 
\begin{example}[Fully non-epistemic, $n=3,d\geq4$] \label{ex_fully_epis}
    Consider three sets MUBs in $\mathbbm{C}^d$, $\left\{\left\{|\psi_{i_k|k}\rangle\right\}_{i_k=1}^{d}\right\}_{k=1}^{3}$ where $d$ is prime-power and $d \geq 4$. Each of these sets will result in a maximally mixed preparation $\rho_k=(1/d)\sum_{i_k=1}^{d}|\psi_{i_k|k}\rangle\!\langle\psi_{i_k|k}|=(\mathbbm{1}/d)$ for $k=1,2,3$ such that $\omega_E^{[3]}(\rho_1,\rho_2,\rho_3)=0$ while $\omega_Q^{[3]}(\rho_1,\rho_2,\rho_3)=1$.
\end{example}

    A prime-power $d$ ensures that at least three MUBs exist because $\mathbbm{C}^d$ where $d$ is prime-power contains a total of $d+1$ MUBs. Since the states are mutually unbiased, $x_1$, $x_2$ and $x_3$, as defined in Corollary \ref{n=3,d=4} are all equal, $x_1=x_2=x_3=1/d$. For $d\geq4$, they satisfy \eqref{caves_eq} and hence $\omega_E^{[3]}(\rho_1,\rho_2,\rho_3)=0$.

\subsection{Sets of two non-maximally epistemic and non-epistemic mixed preparations} \label{sec3_2}

So far, we have discussed non-epistemic and fully non-epistemic instances of sets containing more than two quantum preparations. Now, we divert our attention to the anti-distinguishability of two preparations. Here we would like to point out that for two preparations, anti-distinguishability and distinguishability are the same, and we have used Eqs. \eqref{Lqn2}, \eqref{distin}, and \eqref{distin_pure} wherever required. Apart from focusing on non-epistemic cases, we have also looked into non-maximally epistemic cases, as discussed below. 

\subsubsection{Non-maximally epistemic preparations in $d\geq2$}

\begin{coro} \label{c2}
    Consider any two MUBs in $\mathbbm{C}^d$, $\left\{\left\{|\psi_{i_k|k}\rangle\right\}_{i_k=1}^{d}\right\}_{k=1}^{2}$. Each of them will result in a maximally mixed preparation, $\rho_{k}=(1/d)\sum_{i_k=1}^{d}|\psi_{i_k|k}\rangle\!\langle\psi_{i_k|k}|=(\mathbbm{1}/d)$ for $k=1,2$. Evidently, $\omega_Q(\rho_1,\rho_2)=1$. The following relation holds for their epistemic overlap,
    \begin{equation}\label{wewqd2mubs}
\omega_E(\rho_1,\rho_2)\leq d-\sqrt{d(d-1)} .
    \end{equation}
\end{coro}
\begin{proof}
Substituting $n=2$ in Corollary \ref{c1} we have,
    \begin{equation}
        \omega_E(\rho_1,\rho_2)\leq \frac{2}{d}\sum_{i_1,i_2=1}^{d}\Big(1-D_Q(\psi_{i_1|1},\psi_{i_2|2})\Big).
    \end{equation}
   From \eqref{distin_pure}, $D_Q(\psi_{i_1|1},\psi_{i_2|2})=(1/2)\left(1+\sqrt{1-|\bra{\psi_{i_1|1}}\psi_{i_2|2}\rangle|^2}\right).$ Since we have a set of two MUBs, $|\langle\psi_{i_1|1}|\psi_{i_2|2}\rangle|^2=1/d$ for all $i_1,i_2$, so,
    \begin{eqnarray}
       \omega_E(\rho_1,\rho_2)&\leq& \frac{1}{d}\sum_{i_1,i_2=1}^{d}\left(1-\sqrt{1-\frac{1}{d}}\right)\nonumber \\
       &\leq& \frac{d^2}{d}\left(1-\sqrt{1-\frac{1}{d}}\right),
    \end{eqnarray}
    which implies \eqref{wewqd2mubs}.
\end{proof}
For qubits, the epistemic overlap is upper bounded by $2-\sqrt{2}\approx 0.58$. A natural question that arises is whether one can obtain an epistemic overlap lower than this quantity. As it turns out, $2-\sqrt{2}$ is the minimum epistemic overlap that can be achieved when two maximally mixed preparations acting on $\mathbbm{C}^2$ are concerned. 
\begin{thm} \label{theorem_C2}
 Consider any two complete orthonormal sets in $\mathbbm{C}^2$, $\{\{\ket{\psi_{i_k|k}}\}_{i_k=1}^2\}_{k=1}^2$. Each one of them results in a maximally mixed preparation, $\rho_k=(1/2)\sum_{i_k=1}^2|\psi_{i_k|k}\rangle\!\langle\psi_{i_k|k}|=\mathbbm{1}/2$ for $k=1,2,$ such that the minimum exact value of $\omega_E(\rho_1,\rho_2)$ for all possible ontological models is $2-\sqrt{2}$.

\end{thm}
 \begin{proof}
      Redoing some of the steps in Corollary \ref{c1} makes it clear why. Consider the epistemic overlap of $\rho_1$ and $\rho_2$,
     \bea
     && \omega_E(\rho_1,\rho_2)\nonumber \\
     &=& \int_{\Lambda}\min(\mu(\lambda|\rho_1),\mu(\lambda|\rho_2))d\lambda \nonumber \\
         &=& \frac{1}{2}\int\limits_{\Lambda}\min\Big(\sum_{i_1=1}^2\mu(\lambda|\psi_{i_1|1}), \sum_{i_2=1}^2\mu(\lambda|\psi_{i_2|2})\Big)d\lambda. 
         \eea 
For a given $k=1,2$, the epistemic states, $\mu(\lambda|\psi_{1|k})$ and $\mu(\lambda|\psi_{2|k})$ have no overlap since $\ket{\psi_{1|k}}$ and $\ket{\psi_{2|k}}$ are orthogonal. This allows us to bisect the $\Lambda$-space, $\Lambda=\Lambda_{\psi_{1|k}}+\Lambda_{\psi_{2|k}}$, where $\Lambda_{\psi_{i_k|k}}$ denotes the support of $\mu(\lambda|\psi_{i_k|k})$ such that $\mu(\lambda|\psi_{i_k|k})>0$ $\forall \lambda \in \Lambda_{\psi_{i_k|k}}$. The above integral then becomes,
      \bea \label{th_6_eq}
         && \omega_E(\rho_1,\rho_2) \nonumber \\
          &=& \frac{1}{2}\int\limits_{\Lambda_{\psi_{1|1}}}\min\left(\mu(\lambda|\psi_{1|1}),\sum_{i_2=1}^2\mu(\lambda|\psi_{i_2|2})\right)d\lambda \nonumber \\
          && +\frac{1}{2}\int\limits_{\Lambda_{\psi_{2|1}}}\min\left(\mu(\lambda|\psi_{2|1}),\sum_{i_2=1}^2\mu(\lambda|\psi_{i_2|2})\right)d\lambda. \nonumber \\
          \eea 
      A similar bifurcation of $\Lambda-$space can be done in terms of the sub-spaces $\Lambda_{\psi_{1|2}}$ and $\Lambda_{\psi_{2|2}}$. And since $\Lambda_{\psi_{i_k|k}}\cap\Lambda=\Lambda_{\psi_{i_k|k}}$, \eqref{th_6_eq} can be finally written as,
\begin{equation} \label{th_6_eq_2}
      \begin{split}
\omega_E(\rho_1,\rho_2)&=\frac{1}{2}\sum_{i_1,i_2=1}^2\int_{\Lambda}\min\left(\mu(\lambda|\psi_{i_1|1}),\mu(\lambda|\psi_{i_2|2})\right)d\lambda\\
       &=\frac{1}{2}\sum_{i_1,i_2=1}^2\omega_E(\psi_{i_1|1},\psi_{i_2|2}).
        \end{split}
    \end{equation}
    Note that \eqref{th_6_eq_2} is almost similar to \eqref{coro1_eq} (for $n=2$), the difference being that \eqref{coro1_eq} is an inequality whereas we have obtained an equality here. Furthermore, what we have obtained here can be extended to any $d$-dimensional Hilbert space, an improvement on the bound set by Corollary \ref{c1}.
    
   We now consider the Kochen-Specker model \cite{kochen1967problem}. The Kochen-Specker model is maximally epistemic as far as two pure qubit states are concerned. Hence, we replace $\omega_E(\psi_{i_1|1},\psi_{i_2|2})$ with $\omega_Q(\psi_{i_1|1},\psi_{i_2|2})$ in \eqref{th_6_eq_2}, which yields,
    \begin{equation} \label{we_rho_wq_psi}
        \omega_E(\rho_1,\rho_2)=\frac{1}{2}\sum_{i_1,i_2=1}^2\omega_Q(\psi_{i_1|1},\psi_{i_2|2}).
    \end{equation}
$\omega_Q(\psi_{i_1|1},\psi_{i_2|2})$ is related to $D_Q(\psi_{i_1|1},\psi_{i_2|2})$ by \eqref{Lqn2}. Using the expression of $D_Q(\psi_{i_1|1},\psi_{i_2|2})$ as given in \eqref{distin_pure}, we have,
\begin{equation} \label{wq_pure}
    \omega_Q(\psi_{i_1|1},\psi_{i_2|2})=1-\sqrt{1-|\langle\psi_{i_1|1}|\psi_{i_2|2}\rangle|^2}.
\end{equation}
 Using \eqref{wq_pure}, we can show that $\omega_E(\rho_1,\rho_2)$ attains its minimum value when the orthonormal sets are MUBs.
Let the states given by $k=1$ be expressed in terms of the states given by $k=2$,
\begin{equation}
    \begin{split}
        &\ket{\psi_{1|1}}=c_1\ket{\psi_{1|2}}+c_2\ket{\psi_{2|2}}\\
        &\ket{\psi_{2|1}}=c_1'\ket{\psi_{1|2}}+c_2'\ket{\psi_{2|2}},
    \end{split}
\end{equation}
where obviously $|c_1|^2+|c_2|^2=1$ and $|c_1'|^2+|c_2'|^2=1$. Substituting the expression for $\omega_Q(\psi_{i_1|1},\psi_{i_2|2})$ from \eqref{wq_pure} in \eqref{we_rho_wq_psi} and using the constraints on the coefficients we can write $\omega_E(\rho_1,\rho_2)$ as,
\begin{equation} \label{th_6_omega_c}
\begin{split}
        \omega_E(\rho_1,\rho_2)=2-\frac{1}{2}&\Big(\sqrt{1-|c_1|^2}+|c_1|\\
        &+\sqrt{1-|c_1'|^2}+|c_1'|\Big).
        \end{split}
\end{equation}
To minimize \eqref{th_6_omega_c}, we set $\partial \omega_E(\rho_1,\rho_2)/\partial |c_1|=0$ and $\partial \omega_E(\rho_1,\rho_2)/\partial |c_1'|=0$ and it becomes evident that minimization is achieved when $|c_1|=|c_2|=|c_1'|=|c_2'|= 1/\sqrt{2}$. Clearly, the two orthonormal basis sets must be mutually unbiased.

Using these values, the lowest value that $\omega_E(\rho_1,\rho_2)$ can take in the Kochen-Specker model is $2-\sqrt{2}$. So, for any other basis sets, $\omega_E(\rho_1,\rho_2) > 2-\sqrt{2}$.  
    
    This result is interesting because Corollary \ref{c2} claims that for two MUBs in $d=2$, $\omega_E(\rho_1,\rho_2)\leq 2-\sqrt{2}$, irrespective of any particular epistemic model. Both of these results can be reconciled if one realizes that in $d=2$, since we have shown the existence of an epistemic model where $\omega_E(\rho_1,\rho_2)= 2-\sqrt{2}$, then we cannot have any other model where $\omega_E(\rho_1,\rho_2)< 2-\sqrt{2}$, as that will violate Corollary \ref{c2}. Hence the lowest epistemic overlap possible for two maximally mixed qubit preparations is $2-\sqrt{2}$. Although it must have been clear from Theorem \ref{cond_non_epis_thm}, we emphasize here the fact that the Kochen-Specker model is not maximally epistemic for mixed qubit preparations. 
    \end{proof}
    
\subsubsection{Asymptotically non-epistemic and fully non-epistemic preparations}
    Now that we have discussed the non-maximally epistemic nature of two quantum preparations, we divert our attention to the existence of non-epistemic cases. Such preparations do exist, and to see that we have constructed the following theorem presented below.

\begin{thm} \label{th_bran}
 Consider a set of $d(m+1)$ number of pure states, which are further subdivided into $m+1$ subsets, each containing $d$ number of states, $\{|\psi_{i}\rangle\}_{i=1}^{d}$ and $\{\{|\phi_{i_k|k}\rangle\}_{i_k=1}^{d}\}_{k=1}^{m}$. We construct mixed preparations out of these sets,
    \begin{equation*} \label{rhok}
    \rho_0 = \frac{1}{d}\sum_{i=1}^{d}|\psi_{i}\!\rangle\!\langle\psi_{i}| , \textrm{ } \rho_k = \frac{1}{d} \sum_{i_k=1}^{d}|\phi_{i_k|k}\rangle\!\langle\phi_{i_k|k}| \textrm{ } \forall k \in \{1,\cdots,m\}.
    \end{equation*}
    Then the mixed preparations satisfy the following inequality, 
    \begin{equation}\label{th_5_eq}
        \sum_{k=1}^m \w_E(\rho_0,\rho_k) \leq 1 + \frac{3}{2d} \sum_{j,i_k,i'_{k'}=1}^d \left( 1 - A_Q^{[3]}\left(\psi_j,\phi_{i_k|k},\phi_{i'_{k'}|k'}\right) \right).
    \end{equation}
\end{thm}
\begin{proof}
    Note that we can write 
    \be \label{we_leq_we}
    \sum_{k=1}^m \w_E(\rho_0,\rho_k) \leq \frac{1}{d} \sum_{j=1}^d \sum_{k=1}^m \sum_{i_k=1}^d \w_{E}(\psi_{j},\phi_{i_k|k}),
    \ee 
    which follows from Theorem \ref{t1} for $n=2$, where $p_{i_k|k}=1/d$, after summing over $k$.
    To prove \eqref{th_5_eq}, we start with $\omega_E(\psi_j,\phi_{i_k|k})$ and sum over the indices $k$ and $i_k$,
    \begin{widetext}
    \bea \label{th_7_eq}
    \sum_{k=1}^m\sum_{i_k=1}^d \w_E(\psi_j,\phi_{i_k|k}) &=&  \int_{\Lambda} \sum_{k=1}^m\sum_{i_k=1}^d  \min( \mu(\lambda|\psi_j), \mu(\lambda|\phi_{i_k}|k)) d\lambda \nonumber \\ 
    &\leq & \int_{\Lambda} \max_{k,i_k} \left( \min(\mu(\lambda|\psi_j),\mu(\lambda|\phi_{i_k|k})) \right) d\lambda + \frac12 \int_{\Lambda} \sum_{i_k,i'_{k'}=1}^d \min(\mu(\lambda|\psi_j),\mu(\lambda|\phi_{i_k|k}),\mu(\lambda|\phi_{i'_{k'}|k'}) ) d\lambda  \nonumber \\
    &\leq & 1+\frac{1}{2}\sum_{i_k,i'_{k'}=1}^d\omega_Q^{[3]}\left(\psi_j,\phi_{i_k|k},\phi_{i'_{k'}|k'}\right) \nonumber\\
&= & 1 + \frac32 \sum_{i_k,i'_{k'}=1}^d \left( 1 - A_Q^{[3]}\left(\psi_j,\phi_{i_k|k},\phi_{i'_{k'}|k'}\right) \right) .
\eea 
 \end{widetext}
In the second line, we use the fact that for any set of non-negative numbers $\{a_i\}$,
\be \label{relation1}
\sum_i a_i \leq \max(a_1,\cdots,a_n) + \frac12 \sum_{i,j} \min(a_i,a_j) ,
\ee 
followed by the fact that the second term in the second line is $\sum_{i_{k},i'_{k'}}\omega_E^{[3]}(\psi_j,\phi_{i_k},\phi_{i'_{k'}|k'})$. Also, note that the first term in the second line is always less than or equal to $1$. From there on, arriving at the final expression in \eqref{th_7_eq} is quite straightforward. To finally arrive at \eqref{th_5_eq}, we take \eqref{th_7_eq}, and plug it in \eqref{we_leq_we}. This completes the proof.
\end{proof}

Non-epistemic examples follow from this theorem. It will become apparent through the two theorems presented in the next two sections, which have been modeled after the results of Barrett \textit{et al.} \cite{Barrett}  and Branciard \cite{Branciard} respectively.

\begin{thm}[Fully non-epistemic for two preparations] \label{coro_3}
    Consider the states in $d+1$ number of MUBs present in $\mathbb{C}^d$, $\{|\psi_{i}\rangle\}_{i=1}^{d}$ and $\left\{\left\{|\phi_{i_k|k}\rangle\right\}_{i_k=1}^{d}\right\}_{k=1}^{d}$ where $d$ is prime-power and $d\geq 4$. The maximally mixed preparations constructed out of them,
    \begin{equation*} \label{rhok}
    \begin{split}
    \rho_0 &= \frac{1}{d}\sum_{i=1}^{d}|\psi_{i}\!\rangle\!\langle\psi_{i}|=\frac{\mathbbm{1}}{d}, \\
    \rho_k &= \frac{1}{d} \sum_{i_k=1}^{d}|\phi_{i_k|k}\rangle\!\langle\phi_{i_k|k}|=\frac{\mathbbm{1}}{d} \textrm{ } \forall k \in \{1,\cdots,d\},
    \end{split}
    \end{equation*}
    satisfy the following relations,
    \be \label{cor_3_ineq}
  \frac{1}{d}\sum_{k=1}^d 
 \w_E(\rho_0,\rho_k) \leq \frac{1}{d} \ ,
\ee 
while $ \w_Q(\rho_0,\rho_k) = 1$ for all $k$.
\end{thm}
\begin{proof}
    This result follows from Theorem \ref{th_bran}. We take the $d+1$ maximally mixed preparations and plug them in \eqref{th_5_eq}, noting that $m=d$. Since all the preparations are maximally mixed, $\omega_Q(\rho_0,\rho_k)=1$ for every $k$. We further note that sets such as $\{\ket{\psi_j},\ket{\phi_{i_k|k}},\ket{\phi_{i'_{k'}|k'}}\}$ are perfectly anti-distinguishable, that is $A_Q^{[3]}(\psi_j,\phi_{i_k|k},\phi_{i'_{k'}|k'})=1$ since the states contained in them satisfy \eqref{caves_eq} mentioned in Corollary \ref{n=3,d=4}. Then \eqref{th_5_eq} dictates that $\sum_{k=1}^d\omega_E(\rho_0,\rho_k)\leq1$. Equation \eqref{cor_3_ineq} follows when this quantity is divided by $d$, which signifies the average of $\omega_E(\rho_0,\rho_k)$.
\end{proof}
It is worth mentioning that \eqref{cor_3_ineq} is a better bound than that obtained for pure states by Barrett \textit{et al.} in \cite{Barrett}.

To see how fully non-epistemic instances follow from this theorem, consider a Hilbert space of the form, $\bigotimes_{i=1}^n(\mathbbm{C}^p)_i$, where $p$ is any prime number. The dimension of this Hilbert space $d=p^n$, and hence is prime-power. In the asymptotic limit $n \rightarrow \infty$, the upper bound in \eqref{cor_3_ineq} of Theorem \ref{rhok} approaches 0. Since every term in this average is a non-negative quantity, this implies that as the average approaches $0$, so does at least some of the individual terms, meaning $\omega_E(\rho_0,\rho_k) \rightarrow 0$ for some $k$. Since the preparations are all maximally mixed, $\omega_Q(\rho_0,\rho_k)=1$, which indicates that asymptotically it approaches a fully non-epistemic case.


\begin{thm}[Non-epistemic for two preparations] \label{coro_bran}
    For any Hilbert space $\mathbbm{C}^d$ of dimension $d\geq4$, one can find $nd+1$ number of pure states: $\ket{\psi_0}$ and $\left\{\{\ket{\phi_{i_k|k}}\}_{i_k=1}^d\right\}_{k=1}^n$, such that the pure state $\rho_0=\ket{\psi_0}\!\bra{\psi_0}$ and mixed preparations $\rho_k=(1/d)\sum_{i_k=1}^d\ket{\phi_{i_k|k}}\!\bra{\phi_{i_k|k}}$, $\forall k \in \{1,\cdots,n\}$ satisfy the inequalities, 
    \begin{equation} \label{bran_d_4}
        \frac{1}{n}\sum_{k=1}^n \omega_E(\rho_0,\rho_k) \leq \frac{1}{n} \ ,
    \end{equation}
    and
    \begin{equation} \label{non_epis_th8} \omega_Q(\rho_0,\rho_k) \geq \frac{1}{8(nd)^\frac{1}{d-2}}, \ \forall k.
    \end{equation} 
    
\end{thm}

\begin{proof}
  Branciard has shown in \cite{Branciard} that for any $d\geq3$ and $N>2$, one can always find pure states $\ket{\psi_0}$ and $\{\ket{\phi_i}\}_{i=1}^N$ such that:
 \begin{enumerate}[label=(\roman*)]
     \item $|\bra{\psi_0}\phi_i\rangle|^2\leq \chi=\frac{1}{4}N^{\frac{-1}{d-2}}$
     \item  A set of this form $\{\ket{\psi_0},\ket{\phi_i},\ket{\phi_j}\}$, for all $i\neq j$ is perfectly anti-distinguishable.
 \end{enumerate}
Let $N=nd$, meaning we have $nd+1$ number of states in $\mathbbm{C}^d$ that satisfy the two conditions mentioned above. Furthermore, let us divide the $nd$ number of states into $n$ different sets having $d$ states each, akin to the theorem statement. We then take the $n+1$ mixed preparations arising out of them and plug these in \eqref{th_5_eq} of Theorem \ref{th_bran}, noting that $m=n$. Each such set $\{\ket{\psi_0},\ket{\phi_{i_k|k}},\ket{\phi_{i'_{k'}|k'}}\}$ is likened to sub-part (ii) of Branciard's result as mentioned above, and thus they are perfectly anti-distinguishable. Then according to \eqref{th_5_eq}, $\sum_{k=1}^n\omega_E(\rho_0,\rho_k)\leq1$. The remainder of our task is to find a bound for $\omega_Q(\rho_0,\rho_k)$.

Let us recall that $\omega_Q(\rho_0,\rho_k)$ is related to $D_Q(\rho_0,\rho_k)$ by \eqref{Lqn2}, which in turn is expressed in terms of $T(\rho_0,\rho_k)$ by \eqref{distin}. If one splits the mixed preparations into their pure states, then we have a useful identity at hand,
  \begin{equation}
      T(\rho_0,\rho_k)\leq\frac{1}{d}\sum_{i_k=1}^dT\left(\psi_0,\phi_{i_k|k}\right).
    \end{equation}
      This in turn, implies that,
      \begin{equation}
       D_Q(\rho_0,\rho_k)\leq \frac{1}{d}\sum_{i_k=1}^d D_Q(\psi_0,\phi_{i_k|k}).
  \end{equation}
  We use this identity to construct the following lower bound of $\omega_Q(\rho_0,\rho_k)$,
  \bea \label{w_q_bnd}
          \omega_Q(\rho_0,\rho_k)&=& 2(1-D_Q(\rho_0,\rho_k))\nonumber \\
          & \geq & 2\left(1-\frac{1}{d}\sum_{i_k=1}^dD_Q(\psi_0,\phi_{i_k|k})\right).
      \eea 
  $D_Q(\psi_0,\phi_{i_k|k})=(1/2)\left(1+\sqrt{1-|\bra{\psi_0}\phi_{i_k|k}\rangle|^2}\right)$, as expressed in \eqref{distin_pure}. We substitute $\chi$ from sub-part (i) of Branciard's result in place of $|\bra{\psi_0}\phi_{i_k|k}\rangle|^2$ (keeping in mind that $N=nd$) and use \eqref{w_q_bnd} to arrive at the inequality,
  \begin{equation} \label{sum_wq}
      \begin{split}
          \omega_Q(\rho_0,\rho_k) \geq 2-\left(1+\sqrt{1-\chi}\right) \geq \frac{\chi}{2}=\frac{(nd)^\frac{-1}{d-2}}{8}.
      \end{split}
  \end{equation}
  The above inequality \eqref{sum_wq} when rearranged properly is nothing but \eqref{non_epis_th8}, which completes the proof.
  
\end{proof}

To see how such preparations are asymptotically non-epistemic, consider the fact that the average of $\omega_E(\rho_0,\rho_k)\leq 1/n$ which implies that at least for one value of $k$, $\omega_E(\rho_0,\rho_k)\leq 1/n$. For that $k$ the epistemic overlap compared to the quantum overlap $\omega_E(\rho_0,\rho_k)/\omega_Q(\rho_0,\rho_k)\leq8d^{1/d-2}/n^{d-3/d-2}$. Thus in any Hilbery space of $d\geq4$, for  larger sets of such preparations, that is, as $n\rightarrow\infty$ epistemic models become exceedingly poor at explaining their distinguishability.

\subsubsection{Non-maximally epistemic mixed preparations in the simplest preparation contextuality scenario} \label{sec5}

Let us now consider the expression of the simplest preparation noncontextuality inequality, which also appears as the success metric of parity oblivious multiplexing task \cite{Spekkens2009,PuseyPRA2018}. 
Below, we have presented a theorem that can used to witness non-maximally epistemic mixed preparations, in general, whose specific instances manifest violations of that same preparation noncontextuality inequality. This theorem is closely related to the Proposition 2 in \cite{cs}.

\begin{thm} \label{theorem_S}
    Consider the set of four states $\{\ket{\psi_{x_0x_1}}\}$ where $x_0x_1=00,01,10,11$, two 
    binary-outcome measurements $\{\mathcal{M}_0,\mathcal{M}_1\}$, and the following empirical quantity,
\be 
S = \frac{1}{8} \sum_{x_0x_1,y} p(x_y|\psi_{x_0x_1},\mathcal{M}_y).
\ee 
 Then, the following relation holds,
\be \label{contex_eq}
\omega_E(\rho_0,\rho_1) \leq 4(1-S),
\ee
where 
\bea
\rho_0 = \frac12 (\ket{\psi_{00}}\!\bra{\psi_{00}}+\ket{\psi_{11}}\!\bra{\psi_{11}}), \nonumber \\ \rho_1 = \frac12 (\ket{\psi_{01}}\!\bra{\psi_{01}}+\ket{\psi_{10}}\!\bra{\psi_{10}}).
\eea 

\end{thm}
\begin{proof}
    We begin by writing the probabilities in $S$ in terms of epistemic states and response functions, as given by \eqref{prob},
    \begin{equation}
        S=\frac{1}{8}\sum_{x_0,x_1,y}\int_{\Lambda}\mu(\lambda|x_0x_1)\xi(x_y|\lambda,\mathcal{M}_y)d\lambda.
    \end{equation}
    Next, we write the above equation explicitly and then collect the terms as follows,
    \begin{equation} \label{wewqS}
        \begin{split}
            S=\frac{1}{8}\int_{\Lambda}\Big(&\xi(0|\lambda,\mathcal{M}_0)\big(\mu(\lambda|00)+\mu(\lambda|01)\big)\\ +&\xi(1|\lambda,\mathcal{M}_0)\big(\mu(\lambda|10)+\mu(\lambda|11)\big)\\ +&\xi(0|\lambda,\mathcal{M}_1)\big(\mu(\lambda|00)+\mu(\lambda|10)\big)\\
            +&\xi(1|\lambda,\mathcal{M}_1)\big(\mu(\lambda|11)+\mu(\lambda|01)\big)\Big)d\lambda.
        \end{split}
    \end{equation}
    Given a measurement $\mathcal{M}_y$ for $y\in\{0,1\}$, the response functions  are normalized, $\xi(0|\lambda,\mathcal{M}_y)+\xi(1|\lambda,\mathcal{M}_y)=1$. Hence, one can write,
    \begin{equation}
    \begin{split}
        S\leq\frac{1}{8}&\int_{\Lambda}\Big(\max\big(\mu(\lambda|00)+\mu(\lambda|01),\mu(\lambda|10)+\mu(\lambda|11)\big)\\
        +&\max\big(\mu(\lambda|00)+\mu(\lambda|10),\mu(\lambda|11)+\mu(\lambda|01)\big)\Big)d\lambda.
        \end{split}
    \end{equation}
    Next we use the identity, $\max(a,b)=a+b-\min(a,b)$ which leads us to the inequality given below,
    \begin{equation}
        \begin{split}
            S \leq 1-\frac{1}{8}&\int\limits_{\Lambda}\Big(\min(\mu(\lambda|00)+\mu(\lambda|01),\mu(\lambda|10)+\mu(\lambda|11))\\
        +&\min(\mu(\lambda|00)+\mu(\lambda|10),\mu(\lambda|11)+\mu(\lambda|01))\Big)d\lambda.
        \end{split}
    \end{equation}
    For a set of positive numbers $\{a,b,c,d\}$, $\min(a+b,c+d)+\min(a+c,b+d)\geq\min(a+d,b+c)$ holds true. We use this to finally arrive at,
    \begin{equation}\label{S}
    \begin{split}
        S\leq 1-\frac{1}{4}\int\limits_{\Lambda}\frac{1}{2}\min&\Big(\mu(\lambda|00)+\mu(\lambda|11),\\
        &\mu(\lambda|01)+\mu(\lambda|10)\Big)d\lambda.
        \end{split}
    \end{equation}
    Notice that $(1/2)(\mu(\lambda|00)+\mu(\lambda|11))=\mu(\lambda|\rho_0)$ and $(1/2)(\mu(\lambda|01)+\mu(\lambda|10))=\mu(\lambda|\rho_1)$ and that the quantity sitting on the right-hand-side of \eqref{S} is $1-(1/4)\omega_E(\rho_0,\rho_1)$. Rearranging the terms leads us to \eqref{contex_eq}, which completes the proof.
    
\end{proof}

 The qubit states and measurements achieving the maximal value of $S = \frac{1}{2}\left(1+\frac{1}{\sqrt{2}}\right)$ are provided in \cite{Spekkens2009}, which makes $\omega_E(\rho_0,\rho_1)\leq 2-\sqrt{2}$. For these states, $\omega_Q(\rho_0,\rho_1)=1$, since both the preparations are maximally mixed. Here we recall that Theorem \ref{theorem_C2} has shown that the minimum epistemic overlap achievable for two maximally mixed qubit preparations is $2-\sqrt{2}$.  However, it is possible to identify instances where the upper bound on $\omega_E(\rho_0,\rho_1)$ is less than one while simultaneously also having $\omega_Q(\rho_0,\rho_1) < 1$ \cite{cs}. These instances represent non-maximally epistemic mixed preparations but do not demonstrate preparation contextuality.

\section{Discussions and Conclusions} \label{sec_conclusion} 

In a nutshell, this work addresses the fundamental question of whether any epistemic model can replace quantum theory and provides answers to several unresolved questions. First, it is shown that no epistemic model for mixed preparations, even for qubit systems, can reproduce the empirical predictions of quantum theory. Second, we present the most potent manifestation of refuting epistemic models, where the epistemic overlap for a set of preparations vanishes while the respective quantum overlap reaches its maximum value. Unlike all the previous results, these findings involve a finite number of preparations and do not require any nontrivial assumptions about the underlying epistemic models. Moreover, we demonstrate the existence of two indistinguishable mixed preparations such that their epistemic overlap goes to zero, signifying the most compelling form of preparation contextuality.


In particular, we have shown that there exist sets of quantum mixed preparations belonging to Hilbert spaces of dimensions ranging from $2$ onwards with no epistemic explanation for their anti-distinguishability. After establishing an upper bound on the common epistemic overlap of mixed preparations in terms of the epistemic overlap of their decompositions (Theorem \ref{t1}), we have dwelt in Hilbert space of dimension $2$ extensively. We presented a set of sufficient conditions for the 'un-anti-distinguishability' of a set of pure qubit states (Lemma \ref{lem_bloch}), and then used that to write down a set of \textit{necessary and sufficient} conditions for the anti-distinguishability of three pure qubit states (Theorem \ref{bloch_anti_d_cond}). The conditions thus presented have a geometric interpretation and, hence, are easier to visualize. 
With the help of Theorem \ref{bloch_anti_d_cond}, we have presented a set of criteria to identify a large class of non-epistemic cases involving three preparations (Theorem \ref{cond_non_epis_thm}) and provided an explicit example. 
Next, we have shown that fully non-epistemic mixed preparations do not exist for qubit systems (Theorem \ref{thm_no_fully_epis}). They do exist, however, in higher dimensions, as has been shown by utilizing corollaries (\ref{n=4,d=3} and \ref{n=3,d=4}). Specifically, Corollary \ref{n=4,d=3} has been used to show that there exist sets of four preparations in $d=3$ and Corollary \ref{n=3,d=4} to show that there exist sets of three preparations in $d\geq4$ which are fully non-epistemic. We have also considered sets of two mixed preparations and looked at their distinguishability. We have shown instances of non-maximally epistemic cases (Corollary \ref{c2} and Theorem \ref{theorem_C2}) and then moved on in our search for non-epistemic ones. Theorems \ref{coro_3} and \ref{coro_bran}, which follow from Theorem \ref{th_bran} have been utilized to show non-epistemic cases in their respective asymptotic limits, meaning that Theorem \ref{coro_3} leads to fully non-epistemic cases in the limit where the dimension of the Hilbert space $d \rightarrow \infty$ and Theorem \ref{coro_bran} leads to non-epistemic cases in Hilbert spaces of $d\geq4$ but where the number of preparations $n \rightarrow \infty$.  Interestingly, any proof of preparation contextuality implies that the respective mixed preparations are \textit{non-maximally epistemic}, a weaker version of non-epistemic where the epistemic overlap is required to be less than the quantum overlap. Theorem \ref{theorem_S}, which provides a relation between the epistemic overlap and the success metric of the parity-oblivious multiplexing task, has been specifically used to demonstrate the implication explicitly.
However, a more profound consequence of it is the presence of fully non-epistemic cases for two mixed preparations or fully preparation contextuality, wherein the epistemic overlap is zero for two indistinguishable mixed preparations. The no-go results on the epistemic models are summarized in Table \ref{tab:sumup}. 

\begin{widetext}

\begin{table}[h!]
\centering
\begin{tabular}{|c|c|c|}
\hline
No-go Results & Quantum and Epistemic overlap & Dimension ($d$) \\
\hline \hline
Non-maximally epistemic 
& $\omega_Q(\rho_1,\rho_2) > \omega_E(\rho_1,\rho_2)$ 
& $d=2$ \\
\hline
Non-epistemic 
& $\omega^{[3]}_Q(\rho_1,\rho_2,\rho_3) = 0.116,\; \omega^{[3]}_E(\rho_1,\rho_2,\rho_3)=0$ 
& $d=2$ \\
\hline
Fully non-epistemic 
& $\omega^{[4]}_Q(\rho_1,\rho_2,\rho_3,\rho_4)=1,\; \omega^{[4]}_E(\rho_1,\rho_2,\rho_3,\rho_4)=0$ 
& $d=3$ \\
\hline
Fully non-epistemic 
& $\omega^{[3]}_Q(\rho_1,\rho_2,\rho_3)=1,\; \omega^{[3]}_E(\rho_1,\rho_2,\rho_3)=0$ 
& $d=4$ \\
\hline
Fully non-epistemic for two preparations
& $\omega_Q(\rho_1,\rho_2)=1,\; \omega_E(\rho_1,\rho_2)=0$ 
& $d \to \infty$ \\
(Fully preparation contextual) & & \\
\hline
\end{tabular}
\caption{The main no-go results for epistemic models of mixed-state preparations are summarized above. 
Beyond these no-go results, we also show that fully non-epistemic proofs are not possible for qubit systems, that is, there exists a model for qubit such that $\omega^{[n]}_Q(\rho_1, \cdots,\rho_n)=1$ implies $\omega^{[n]}_E(\rho_1, \cdots,\rho_n) > 0.$  }\label{tab:sumup}
\end{table}
\end{widetext}

\subsection{Refuting $\psi-$epistemic models} \label{sec4}

Interestingly, a lack of epistemic mixed preparations does not automatically negate the epistemic model of pure preparations, namely, $\psi-$epistemic models. 
However, we identify that some of the examples obtained here indeed fall into this category.
Consider the scenario where a set of two mixed preparations $\{\rho_1,\rho_2\}$ is shown to have,
\be \label{wen0}
\w_E(\rho_1,\rho_2)=0.
\ee 
 The question we address is about the applicability of $\psi-$epistemic models. In other words, is the non-existence of a set of two mixed preparations that are epistemic sufficient to infer conclusions about the effectiveness of the $\psi-$epistemic models? Let each mixed preparation be decomposed as $\rho_k = \sum_{i_k} p_{i_k|k} \ket{\psi_{i_k|k}}\!\bra{\psi_{i_k|k}}$. It follows from convexity that \eqref{wen0} implies 
\be \label{we-psi}
\w_E(\psi_{i_1|1},\psi_{i_2|2})=0
\ee 
for all values of $i_1$ and $i_2$. Simultaneously, if $\w_Q(\psi_{i_1|1},\psi_{i_2|2}) > 0$, 
then we can conclude that no $\psi-$epistemic theory is capable of explaining the quantum overlap.
However, if one takes a look at the \textit{method} employed in Section \ref{sec3_1} to arrive at the non-epistemic and fully non-epistemic results, it becomes clear that method cannot be used here for this purpose. There, we have used various means and showed that $A_Q^{[n]}(\psi_{i_1|1},\cdots,\psi_{i_n|n})=1$ (where $n\geq3$), or equivalently, $\omega_Q^{[n]}(\psi_{i_1|1},\cdots,\psi_{i_n|n})=0$ to arrive at the fact that the set of mixed preparations is non-epistemic or fully non-epistemic. Here, we explicitly require the quantum overlap $\omega_Q(\psi_{i_1|1},\psi_{i_2|2})>0$. Theorem \ref{rhok} on the other hand can be employed to refute $\psi-$epistemic models.

As discussed in Theorem \ref{coro_3}, consider a Hilbert space of $d=p^n$. 
Eq.~\eqref{cor_3_ineq} implies that there is at least one $k$ such that,
\begin{equation} \label{we_less_than_1/d}
    \omega_E(\rho_0,\rho_k)\leq\frac{1}{d}.
\end{equation}
Recall that $\rho_0$ is an equal mixture of $\{\ket{\psi_i}\}_{i=1}^d$ and $\rho_k$ is an equal mixture of $\{\ket{\phi_{j|k}}\}_{j=1}^d$, where both sets are orthonormal basis sets. Consequently, the orthogonality relations between $\{\ket{\psi_i}\}_{i=1}^d$ and $\{\ket{\phi_{j|k}}\}_{j=1}^d$ allow us to write $\omega_E(\rho_0,\rho_k)$ in \eqref{we_less_than_1/d} as,
\begin{equation} \label{we=sum_of_we}
    \omega_E(\rho_0,\rho_k)=\frac{1}{d}\sum_{i,j=1}^d\omega_E(\psi_i,\phi_{j|k}).
\end{equation}
Note that the above equation is a generalization of \eqref{th_6_eq_2} from the proof of Theorem \ref{theorem_C2}. From \eqref{we_less_than_1/d} we know that \eqref{we=sum_of_we} has to be less than $1/d$. Since there are $d^2$ number of terms on the right-hand-side of \eqref{we=sum_of_we}, $\omega_E(\psi_i,\phi_{j|k})\leq1/d^2$ for at least one value of $i$ and $j$. Furthermore, $\{\ket{\psi_i}\}_{i=1}^d$ and $\{\ket{\phi_{j|k}}\}_{j=1}^d$ are mutually unbiased, and hence from \eqref{Lqn2} and \eqref{distin_pure}, we see that $\omega_Q(\psi_i,\phi_{j|k})=1-\sqrt{1-\frac{1}{d}}\geq 1/2d$. Notice that $\omega_E(\psi_i,\phi_{j|k})$ is one order smaller compared to $\omega_Q(\psi_i,\phi_{j|k})$ and in the limit $d\rightarrow \infty$, $\omega_Q(\psi_i,\phi_{j|k})\gg\omega_E(\psi_i,\phi_{j|k})$, implying that $\psi-$epistemic models fail to account for indistinguishability of two mixed preparations.

\subsection{Open questions}

Among many open questions, let us point out a few. Recall that we have demonstrated the existence of fully non-epistemic cases in $d=3$ for $4$ mixed preparations. 
So far, we have been unable to identify instances of fully non-epistemic cases in $d=3$ for three preparations, and hence it remains an open question. Coming to fully non-epistemic cases for two mixed preparations, we have shown their existence using Theorem \ref{coro_3}, in the asymptotic limit where $d \rightarrow \infty$. Whether or not such phenomena occur in finite-dimensional Hilbert spaces or with a finite number of preparations is another interesting question. Let us shed some insight into this. Consider a mixed state $\varrho$ with two different preparations $\rho_1=\sum_i p_i\ket{\psi_i}\!\bra{\psi_i}$ and $\rho_2=\sum_j q_j \ket{\phi_j}\!\bra{\phi_j}$.
In this case, we have
\begin{equation}
 0 < \tr(\varrho^2) = \tr(\rho_1\rho_2)=\sum_{i,j}p_iq_j |\bra{\psi_i}\phi_j\ra|^2 .
\end{equation}
The above relation implies that there exists at least one such pair of states, denoted as $\ket{\psi_i}$ and $\ket{\phi_j}$, which are non-orthogonal. On the other hand, for the set 
$\{\rho_1,\rho_2\}$ to be fully non-epistemic $\omega_E(\rho_1,\rho_2)$ has to be $0$, which, according to Theorem \ref{t1} necessitates $\omega_E(\psi_i,\phi_j)=0$ for all pairs of $i$ and $j$. 
However, as per \cite{Lewis}, an epistemic model exists for any pair of states, $\ket{\psi_i}$ and $\ket{\phi_j}$, such that $|\bra{\psi_i}\phi_j\ra|^2>(d-1)/d$, where $d$ represents the dimension of the Hilbert space. Thus, $\{\rho_1,\rho_2\}$ cannot be fully non-epistemic whenever at least one pair $\ket{\psi_i}$ and $\ket{\phi_j}$ satisfying this condition exists. Nonetheless, exploring cases where this condition is violated would warrant further investigation. 

 One can also look into cases where the prior probabilities with which the preparations appear in distinguishability and anti-distinguishability are not uniform, but arbitrary. Shin \textit{et al.} reported in \cite{Shin2021QuantumContextual} the dependence of contextual advantage in state discrimination of mixed states on prior probabilities, it would be interesting to see if similar phenomena or other novel features can be observed in our case as well.  Another direction for future investigation involves determining whether a given set of quantum mixed states can be epistemic or not by considering all their possible realizations. Identifying the entire sets of mixed preparations that are non-epistemic and fully non-epistemic could be an interesting research endeavour. Given the implications of previous no-go results on epistemic models in various information-theoretic applications \cite{app1,app2,app3}, exploring the findings presented in this study for quantum advantage in information processing holds considerable interest.


\subsection*{Acknowledgment}
DS acknowledges the financial support from STARS (Grant no. STARS/STARS-2/2023-0809), Govt. of India.


\appendix

\section{Proof of Theorem \ref{t1}}\label{app:pot1}

The proof of Theorem \ref{t1} depends on an inequality involving non-negative numbers presented below.

\begin{lemma}\label{lemma1}
Consider a set of $nr$ many non-negative numbers, $a_{i_k|k}$ that are further divided into $n$ subsets, each having $r$ number of elements in them, denoted as $\left\{\{a_{i_k|k}\}_{i_k=1}^{r}\right\}_{k=1}^{n}$. Naturally, $r,n \in \mathbbm{N}$. Then the elements of these sets satisfy the following inequality,
\bea \label{eq:lemma}    
&& \min\left(\sum_{i_1=1}^{r}a_{i_1|1},\sum_{i_2=1}^{r}a_{i_2|2},\cdots,\sum_{i_n=1}^{r}a_{i_n|n}\right) \nonumber \\
&\leq & \sum_{i_1,i_2,\cdots,i_n=1}^{r}\min\left(a_{i_1|1},a_{i_2|2},\cdots,a_{i_n|n}\right).
\eea 
Moreover, if the left-hand-side of \eqref{eq:lemma} is zero, then the inequality becomes equality, that is, each term on the right-hand-side is zero.
\end{lemma}
\begin{proof}
Without loss of generality, let 
\be \min\left(\sum_{i_1=1}^{r}a_{i_1|1},\sum_{i_2=1}^{r}a_{i_2|2},\cdots,\sum_{i_n=1}^{r}a_{i_n|n}\right)=\sum_{i_1=1}^{r}a_{i_1|1}.
\ee 
We take $\sum_{i_1=1}^{r}a_{i_1|1}$, and subtract it from the quantity in the right-hand-side of \eqref{eq:lemma},
\begin{equation} \label{lemm1_sum}
\sum_{i_1=1}^{r}\left(\sum_{i_2,i_3,\cdots,i_n=1}^{r}\min\left(a_{i_1|1},a_{i_2|2},\cdots,a_{i_n|n}\right)-a_{i_1|1}\right).
\end{equation}
Since $\sum_{i_1=1}^{r}a_{i_1|1}\leq\sum_{i_k=1}^{r}a_{i_k|k}$ for every $k$, and every $a_{i_k|k}\geq0$, it is easy to see that   $a_{i_1|1}\leq\sum_{i_k=1}^{r}a_{i_k|k}$ for every $i_1$ and $k$. Given this fact, it can be shown that for every $i_1$, 
\begin{equation} \label{central_ineq_lem1}
\sum_{i_2,i_3,\cdots,i_n=1}^{r}\min\left(a_{i_1|1},a_{i_2|2},\cdots,a_{i_n|n}\right)-a_{i_1|1} \geq0.
\end{equation}
This would imply that \eqref{lemm1_sum} is greater than or equal to $0$, and thus, \eqref{eq:lemma} will be proved. Moreover if $\sum_{i_1=1}^{r}a_{i_1|1}=0$, it implies that $a_{i_1|1}=0$ for every $i_1$, and in such a scenario, \eqref{eq:lemma} becomes an equality. 

Why \eqref{central_ineq_lem1} holds is not immediately apparent. Consider the following scenarios. First, consider the case where there exists at least one term in the expression,
\be \label{exlem1}
\sum_{i_2,i_3,\cdots,i_n=1}^{r}\min\left(a_{i_1|1},a_{i_2|2},\cdots,a_{i_n|n}\right),
\ee 
such that $\min\left(a_{i_1|1}, a_{i_2|2},\cdots,a_{i_n|n}\right)=a_{i_1|1}$. In that case, all the terms in \eqref{central_ineq_lem1} are clearly greater than or equal to $0$. Next, we consider the case where none of the terms in \eqref{exlem1} yield $a_{i_1|1}$. In that scenario, \eqref{central_ineq_lem1} can  be said to be greater than or equal to $0$ if \eqref{exlem1} contains $\sum_{i_k=1}^{r}a_{i_k|k}$ for at least one $k$, since $\sum_{i_k=1}^{r}a_{i_k|k}\geq a_{i_1|1}$. The crux of the proof lies in showing that this always holds. 

We begin by assuming the contrary, which is to say that we assume \eqref{exlem1} does not contain $\sum_{i_k=1}^{r}a_{i_k|k}$ for all $k\neq1$. To be precise, this implies that for every $k\neq1$, \eqref{exlem1} must be devoid of at least one $a_{i_k|k}$. Without loss of generality, let the missing elements be $a_{r|k}$. To realize this, we divide the total number of terms in \eqref{exlem1} which is $r^{n-1}$, into smaller batches. Every term in each batch yields $a_{i|k}$ for a fixed value of $i$ and $k$. The first batch contains terms of the form, $\min\left(a_{i_1|1},a_{1|2},a_{i_3|3},\cdots,a_{i_n|n}\right)$, all of which are equal to $a_{1|2}$. Thus, the first batch contains $r^{n-2}$ number of terms. Once we have exhausted this batch, none of the terms left in \eqref{exlem1} can yield $a_{1|2}$. Since we aim to prevent forming a sum like $\sum_{i_k=1}^{r}a_{i_k|k}$, we maximize the number of times $a_{i|k}$ for a fixed $i$ and $k \in \{2,\cdots,n\}$ can appear in \eqref{exlem1}, before moving on to the $i+1$th term. In other words, we move through the batches such that each one of them yields $a_{1|k}$, before moving on to the batches which yield $a_{2|k}$, then $a_{3|k}$ and so on, up to $a_{r-1|k}$. At the end of this process, the total number of terms must be equal to $r^{n-1}$ to show that none of the $a_{r|k}$ have appeared. Note that the order in which we move through the batches is not unique but has been chosen for convenience. For the sake of clarity, we denote the number of times a particular element $a_{i|k}$ appears in \eqref{exlem1} with $\#a_{i|k}$. We have already seen that, $\#a_{1|2}=r^{n-2}$. Next, we see that $\#a_{1|3}=r^{n-2}-r^{n-3}=(r-1)r^{n-3}$. We have subtracted $r^{n-3}$ from $r^{n-2}$ to exclude the terms that are present in the first batch and have already yielded $a_{1|2}$ instead of $a_{1|3}$. Continuing in this fashion, it becomes apparent that the number of times $a_{i|k}$ occurs in \eqref{exlem1} is,
\begin{equation}
    \#a_{i|k}=(r-i)^{k-2}(r-i+1)^{n-k}.
\end{equation}
So, the total number of terms is,
\begin{equation}
    \sum_{i=1}^{r-1}\sum_{k=2}^{n}\#a_{i|k}=r^{n-1}-1.
\end{equation}
There is one term missing. Since we have exhausted every element $a_{i|k}$, for all $k\in\{2,\cdots,n\}$ and $i\in\{1,\cdots,r-1\}$, the missing term could only be $a_{r|k}$ for some $k\neq1$. This falsifies the assumption we began with because the presence of this last term completes the sum $\sum_{i_k=1}^{r}a_{i_k|k}$, and thus \eqref{central_ineq_lem1} is always greater than or equal to $0$.
\end{proof}


Lemma \ref{lemma1} is also applicable to sets of non-negative numbers with different number of elements in each set, $\left\{\{a_{i_k|k}\}_{i_k=1}^{r_k}\right\}_{k=1}^{n}$ which are equal to $r_k$ for $k \in \{1,\cdots,n\}$ instead of the same number of elements, $r$ as given in the lemma statement. This is so because we can always choose $r$ as the cardinal number of the set with the maximum number of elements, and add $r-r_k$ number of zeroes in the $k$th set to compensate for that. This does not make a difference since every term on the right-hand-side of \eqref{eq:lemma} which contains $0$ as one of the arguments of $\min(.)$ will not contribute to the sum.

\begin{widetext}
    \begin{proof}[Proof of Theorem \ref{t1}]
    Due to \eqref{fcm}, the epistemic states corresponding to the mixed preparations are given by  
    \be 
    \mu(\lambda|\rho_k)=\sum_{i_k=1}^{m}p_{i_k|k}\mu(\lambda|\psi_{i_k|k}), \quad \forall \lambda \in \Lambda, k \in \{1,\dots,n\}.
    \ee 
    Plugging these in \eqref{epis_overlap} we have,
    \begin{equation}\label{th1_1}
        \omega_E^{[n]}(\rho_1,\cdots,\rho_n)=\int_{\Lambda}\min\left(\sum_{i_1=1}^{m}p_{i_1|1}\mu(\lambda|\psi_{i_1|1}),\cdots,\sum_{i_n=1}^{m}p_{i_n|n}\mu(\lambda|\psi_{i_n|n})\right)d\lambda.
    \end{equation}
    Next, we apply Lemma \ref{lemma1} to (\ref{th1_1}). Once we have done that, we end up with,
    \begin{equation}\label{th1_2}
        \omega_E^{[n]}(\rho_1,\cdots,\rho_n)\leq\sum_{i_1,i_2,\cdots,i_n=1}^m\int_{\Lambda}\min\left(p_{i_1|1}\mu(\lambda|\psi_{i_1|1}),\cdots,p_{i_n|n}\mu(\lambda|\psi_{i_n|n})\right)d\lambda.
    \end{equation}
    Before proceeding further, consider two sets of non-negative numbers, $\{\alpha_1,\cdots,\alpha_n\}$ and $\{a_1,\cdots,a_n\}$. For two such sets, $\min(\alpha_1a_1,\cdots,\alpha_na_n)\leq\max(\alpha_1,\cdots,\alpha_n)\min(a_1,\cdots,a_n)$. The reason is quite straightforward. Let us consider that $\max(\alpha_1,\cdots,\alpha_n)=\alpha_i$ and $\min(a_1,\cdots,a_n)=a_j$ for some fixed value of $i$ and $j$. Now if $\min(\alpha_1a_1,\cdots,\alpha_na_n)=\alpha_ka_k$, then $\alpha_ka_k \leq\alpha_ja_j$. Since $\max(\alpha_1,\cdots,\alpha_n)=\alpha_i$, $\alpha_j \leq \alpha_i$ which implies that $\alpha_ja_j\leq \alpha_ia_j=\max(\alpha_1,\cdots,\alpha_n)\min(a_1,\cdots,a_n)$. Thus, $\min(\alpha_1a_1,\cdots,\alpha_na_n)\leq\max(\alpha_1,\cdots,\alpha_n)\min(a_1,\cdots,a_n)$ holds. Every term being summed over in \eqref{th1_2} can be cast in this inequality with the convex coefficients being identified with $\{\alpha_i\}_i$ and the epistemic states of the decompositions with $\{a_i\}_i$ which results in,
    \begin{equation}
        \omega_E^{[n]}\left(\rho_1,\cdots,\rho_n\right)\leq\sum_{i_1,i_2,\cdots,i_n=1}^m\max(p_{i_1|1},\cdots,p_{i_n|n})\omega_E^{[n]}(\psi_{i_1|1},\cdots,\psi_{i_n|n}),
    \end{equation}
    where we have simply replaced $\int_{\Lambda}\min(\mu(\lambda|\psi_{i_1|1}),\cdots,\mu(\lambda|\psi_{i_n|n}))d\lambda$ with $\omega_E^{[n]}(\psi_{i_1|1},\cdots,\psi_{i_n|n})$ using the definition of $\omega_E^{[n]}$ given in \eqref{epis_overlap}.
    This is the intermediate bound on $\omega_E^{[n]}(\rho_1,\cdots,\rho_n)$ as shown in \eqref{sub_eq1}. To obtain the upper bound, we simply use \eqref{we<wq} in relation to $\omega_E^{[n]}(\psi_{i_1|1},\cdots,\psi_{i_n|n})$ and replace $\omega_E^{[n]}(\psi_{i_1|1},\cdots,\psi_{i_n|n})$ with $\omega_Q^{[n]}(\psi_{i_1|1},\cdots,\psi_{i_n|n})$. Expressing $\omega_Q^{[n]}(\psi_{i_1|1},\cdots,\psi_{i_n|n})$ in terms of $A_Q^{[n]}(\psi_{i_1|1},\cdots,\psi_{i_n|n})$ according to \eqref{Q_overlap} results in the upper bound \eqref{sub_eq2}.
\end{proof} 
\end{widetext}


\section{Proof of Theorem \ref{bloch_anti_d_cond}}\label{app:pot2}

Based on the lemma introduced below, we will present the proof of Theorem \ref{bloch_anti_d_cond}.

\begin{figure}[h]
    \centering
    \includegraphics[width=0.38\textwidth]{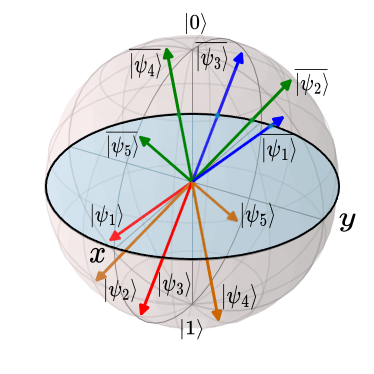}
    \caption{(Lemma \ref{lem_bloch}) Five pure qubit states are situated in the hemisphere $z<0$, and consequently their orthogonal states are situated in $z>0$. The hemispheres are separated by the shaded great circle. Since the $z-$component of every $\ket{\overline{\psi_i}}$ is greater than $0$, there exists no POVM which is capable of perfectly anti-distinguishing $\ket{\psi_1}, \ket{\psi_2}, \ket{\psi_3}, \ket{\psi_4}$ and $\ket{\psi_5}$.}
    \label{fig:bloch_anti_d}
\end{figure}
\begin{lemma} \label{lem_bloch}
   Consider a set of $n$ pure qubit states, $\{\ket{\psi_i}\}_{i=1}^n$. Let their Bloch representations be denoted by $\{\Vec{v'_i}\}_{i=1}^n$. If there exists a unitary operator $U$, whose Bloch representation is the rotation operator $R$ such that, $R\Vec{v'_i}=\Vec{v_i}$, and 
 the same coordinates of the transformed vectors have the same sign, say $(v_z)_i<0$ $\forall i \in \{1,\cdots,n\}$, then $A_Q^{[n]}(\psi_1,\cdots,\psi_n)<1$.
\end{lemma}
\begin{proof}
     For a set of pure states to be perfectly anti-distinguishable, \eqref{anti-d} dictates that there must exist a measurement $\mathcal{M} \equiv \{M_k\}_{k=1}^n$ such that $\sum_{x=1}^np(x|\psi_x,\mathcal{M})=\sum_{x=1}^n\tr(\ket{\psi_x}\!\langle\psi_x|M_x)=0$. Our task will be to show that given the conditions mentioned in the lemma statement, such a measurement does not exist.

        The existence of $U$ indicates that the qubit states $\{\ket{\psi_i}\}_{i=1}^n$ are strictly present in one hemisphere (excluding the base of the hemisphere) of the Bloch sphere since they can be rotated and made to lie in the region $z<0$. Thus, without loss of generality, we consider the states to lie in $z<0$ to surpass the unnecessary task of finding a suitable $U$ (see FIG: \ref{fig:bloch_anti_d}). Their orthogonal states, $\{\ket{\overline{\psi_i}}\}_{i=1}^n$, consequently lie in the region $z>0$. To have perfect anti-distinguishability, without loss of generality, we consider a POVM set such that $M_k=\gamma_k\ket{\overline{\psi_k}}\!\langle\overline{\psi_k}|$ where $\gamma_k\geq0$ (positive semi-definiteness) and $\sum_{k=1}^n\gamma_k\ket{\overline{\psi_k}}\!\langle\overline{\psi_k}|=\mathbbm{1}$ (completeness). $M_k=\frac{\gamma_k}{2}(\mathbbm{1}+\Vec{\sigma}.\Vec{u_k})$, where $\Vec{\sigma}=(\sigma_x,\sigma_y,\sigma_z)$ and $\Vec{u_k}=(u_x,u_y,u_z)_k$ is the Bloch representation of $\ket{\overline{\psi_k}}$. To satisfy the positive-semi-definiteness condition and completeness of the POVM set, we must have,
        \begin{enumerate}[label=(\roman*)]
            \item $\sum_{k=1}^n\gamma_k=2$
            \item $\sum_{k=1}^n\gamma_k(u_z)_k=0.$
        \end{enumerate}
        However since the $z-$component of $\vec{u_k}$ $\forall k \in \{1,\cdots,n\}$ is strictly greater than $0$, (see FIG: \ref{fig:bloch_anti_d}) i.e. $(u_z)_k>0$ the above conditions are incompatible with each other. Thus, there exists no such POVM $\mathcal{M}$ such that $A_Q^{[n]}(\psi_1,\cdots,\psi_n)=1$ implying that $\{\ket{\psi_i}\}_{i=1}^n$ is not perfectly anti-distinguishable. 
\end{proof}

\begin{figure}[h]
\centering
\begin{subfigure}{.20\textwidth}
  \centering
  \includegraphics[width=1\linewidth]{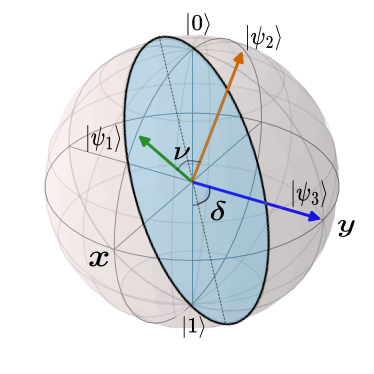}
  \caption{$\ket{\psi_1}$ and $\ket{\psi_2}$ lie on the $Z-X$ plane. $\ket{\psi_3}$ lies on the $X-Y$ plane.}
  \label{fig:sub1}
\end{subfigure}%
\begin{subfigure}{.20\textwidth}
  \centering
  \includegraphics[width=1\linewidth]{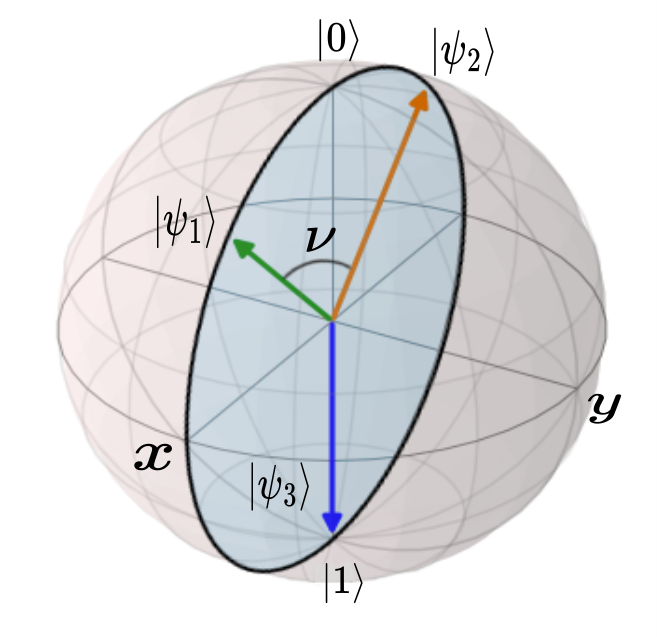}
  \caption{All the states, $\ket{\psi_1}$, $\ket{\psi_2}$ and $\ket{\psi_3}$ lie on the $Z-X$ plane.}
  \label{fig:sub3}
\end{subfigure}%
\caption{(Theorem \ref{bloch_anti_d_cond}) Three pure qubit states $\ket{\psi_1}$, $\ket{\psi_2}$ and $\ket{\psi_3}$ are situated such that in the first diagram, (a) they belong to the same hemisphere. The shaded great circle divides the sphere into two halves. Since they are situated in the same hemisphere, they are not perfectly anti-distinguishable. In (b), they all belong to a single great circle (shaded), which serves as the first step towards the states being perfectly anti-distinguishable.}
\label{fig:bloch}
\end{figure}

\begin{proof}[Proof of Theorem \ref{bloch_anti_d_cond}]
    Consider a great circle of the Bloch sphere that contains the states $|\psi_1\rangle$ and $|\psi_2\rangle$ with an angle $\nu<\pi$ between them. $|\psi_3\rangle$ makes an angle $\delta$ with this great circle. As long as $0<\delta< \pi$, all three states lie in the same hemisphere (excluding the hemisphere's base). One way to visualize this is to imagine an appropriate great circle lying just below $|\psi_3\rangle$ slicing the sphere into two halves (FIG: \ref{fig:sub1}). Lemma \ref{lem_bloch} shows that as long as the three states lie in one hemisphere, they cannot be perfectly anti-distinguished. So, the only recourse is when all the states lie on a great circle (FIG: \ref{fig:sub3}).
    \begin{figure}[h]
\centering
\begin{subfigure}{.20\textwidth}
  \centering
  \includegraphics[width=1\linewidth]{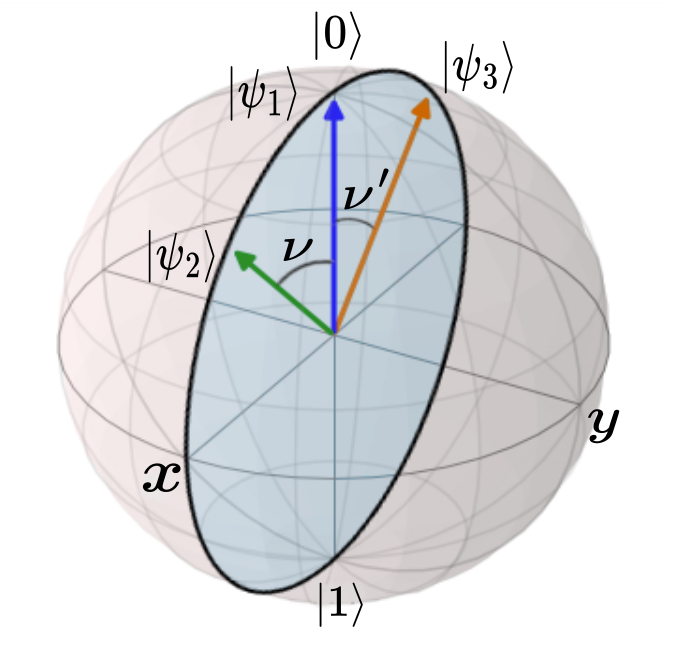}
  \caption{configuration with $\nu+\nu'<\pi$}
  \label{fig:sub_a}
\end{subfigure}%
\begin{subfigure}{.20\textwidth}
  \centering
  \includegraphics[width=1\linewidth]{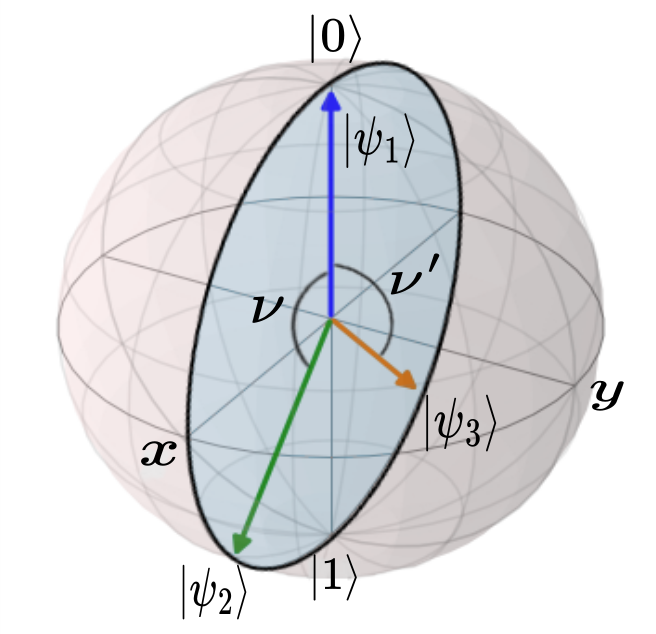}
  \caption{configuration with $\nu +\nu'\geq\pi$}
  \label{fig:sub_b}
\end{subfigure}
\caption{(Theorem \ref{bloch_anti_d_cond}) The three pure qubit states $\ket{\psi_1}$, $\ket{\psi_2}$ and $\ket{\psi_3}$ are situated on the $Z$-$X$ plane (shaded) with $\ket{\psi_1}=\ket{0}$. In the first diagram, (a), their orientation is such that all three states belong to the same semi-circle and hence are not perfectly anti-distinguishable. In the second diagram, (b), the orientation is such that the states are perfectly anti-distinguishable.}
\label{fig:bloch1}
\end{figure}

Without loss of generality, we can consider the great circle containing the states to be the $Z-X$ plane of the Bloch sphere, and $\ket{\psi_1}$ to be $\ket{0}$. 
Let us denote the smaller angle between $|\psi_1\rangle,|\psi_2\rangle$ by $\nu,$ the smaller angle between $|\psi_1\rangle,|\psi_3\rangle$ by $\nu'$, and the smaller angle between $|\psi_2\rangle,|\psi_3\rangle$ by $\nu''$. As long as they all lie on the same semi-circle (by extension, part of the same hemisphere), and hence cannot be perfectly anti-distinguished (FIG: \ref{fig:sub_a}). Thus, to be anti-distinguishable, they should not lie in a semi-circle, which is equivalent to the condition that the sum of every pair of angles between the vectors is greater than or equal to $\pi$, that is, $\nu + \nu' \geq \pi,$ $\nu+\nu'' \geq \pi,$ $\nu'+\nu'' 
\geq \pi.$ By noting that $a=|\langle\psi_1|\psi_2\rangle|=\cos(\nu/2)$, $b=|\langle\psi_1|\psi_3\rangle|=\cos(\nu'/2)$ and $c=|\langle\psi_2|\psi_3\rangle|=\cos(\nu''/2)$, we can arrive at \eqref{bloch_cond} as the necessary condition for anti-distinguishability.

Finally, in order to show that \eqref{bloch_cond} is also sufficient condition, it suffices to provide a POVM of the form, $\mathcal{M} \equiv \{M_1=\gamma_1\ket{\overline{\psi_1}}\!\bra{\overline{\psi_1}},M_2=\gamma_2\ket{\overline{\psi_2}}\!\bra{\overline{\psi_2}},M_3=\gamma_3\ket{\overline{\psi_3}}\!\bra{\overline{\psi_3}}\}$, which can perfectly anti-distinguish them (FIG: \ref{fig:sub_b}). Here,$\ket{\overline{\psi_i}}$ represents the orthonormal state to $\ket{\psi_i}.$
By solving the completeness condition to $\mathcal{M}$ we obtain,
\begin{equation}
    \begin{split}
        &\gamma_1=\frac{-2\sin(\nu+\nu')}{\sin\nu+\sin\nu'-\sin(\nu+\nu')},\\
        &\gamma_2=\frac{2\sin\nu'}{\sin\nu+\sin\nu'-\sin(\nu+\nu')},\\
        &\gamma_3=\frac{2\sin\nu}{\sin\nu+\sin\nu'-\sin(\nu+\nu')},
    \end{split}
\end{equation}
which imply the positive semi-definite condition since $\nu+\nu'\geq \pi$. This completes the proof.
\end{proof}

\section{Proof of Theorem \ref{thm_no_fully_epis}}\label{app:pot3}

To begin, we will introduce a compelling aspect of the  RTQ model that will play a pivotal role in proving Theorem \ref{thm_no_fully_epis}.

\begin{lemma} \label{lemm_KS}
         Consider a set of $n$ pure qubit states, $\{\ket{\psi_i}\}_{i=1}^n$. Let their Bloch representations be denoted by $\{\Vec{v_i}\}_{i=1}^n$ and the angle between any two Bloch vectors $\vec{v_i}$ and $\vec{v_j}$ be $\theta(i,j)$. If $\theta_{max}=\max_{\{i,j\}}(\theta(i,j))<\pi$ then in RTQ model $\omega_E^{[n]}(\psi_1,\cdots,\psi_n)>0$.
    \end{lemma}
    \begin{proof}
       The condition, $\theta_{max}=\max_{\{i,j\}}(\theta(i,j))<\pi$ implies that the states lie in one hemisphere of the Bloch sphere (excluding the hemisphere's base). Let's label this hemisphere as '$I$'. Let there be $m<n/2$ pairs of states, $\{\vec{v'_i},\vec{v'_j}\}_{i,j}$, (they are primed to distinguish them from the rest of the states) such that the angle between each pair is equal to $\theta_{max}$ (see FIG: \ref{fig:full_non_epis}). For each such pair $(\vec{v'_i},\vec{v'_j})$ there exists $\Lambda_{\vec{v'_i},\vec{v'_j}}\subset\mathcal{S}^2$ such that $\Theta(\vec{v}'_i.\vec{r}_1)=\Theta(\vec{v}'_i.\vec{r}_2)=\Theta(\vec{v}'_j.\vec{r}_1)=\Theta(\vec{v}'_j.\vec{r}_2)=1$ for all $\vec{r}_1,\vec{r}_2 \in \Lambda_{\vec{v'_i},\vec{v'_j}}$. Each of these $\Lambda_{\vec{v'_i},\vec{v'_j}}$ corresponds to a spherical wedge situated in hemisphere $I$ with a dihedral angle of $\pi-\theta_{max}$, and are rotated versions of each other about the hemisphere $I$'s axis. Consequently, they have a common overlap region, $C$ (see FIG: \ref{fig:full_non_epis}). That region corresponds to the non-null intersection of $\Lambda_{\vec{v'_i},\vec{v'_j}}$, $\Lambda_C=\bigcap_{\vec{v'_i},\vec{v'_j}}\Lambda_{\vec{v'_i},\vec{v'_j}}$, and we see that,
       \begin{equation}
       \begin{split}
           &\Theta(\vec{v}_i'.\vec{r}_1)=1, \textrm{ and}\\
          &\Theta(\vec{v}_i'.\vec{r}_2)=1 \textrm{ }\ \forall \textrm{ } \vec{r}_1, \vec{r}_2 \in \Lambda_C \textrm{ and }\forall \textrm{ }\vec{v}_i'.
           \end{split}
       \end{equation}

We are halfway there. The remaining task is to show that $\Lambda_C$ has a non-null intersection with $\bigcap_{i=1}^n\Lambda_{\vec{v_i}}$ where every $\Lambda_{\vec{v}_i}\subset \mathcal{S}^2$ is a hemisphere such that its axis is $\vec{v}_i$. Notice that there exists a region $R$ in hemisphere $I$ which contains all the states, and every $\vec{v'_i}$ lies on this region's boundary. The region $C$ (or a portion of it) also lies in $R$. Since each one of these states $\vec{v_i}$ has a corresponding $\Lambda_{\vec{v_i}}$ which is a hemisphere with $\vec{v_i}$ as its axis, all these hemispheres overlap in the region $R$, which in turn overlaps with $C$ (see FIG: \ref{fig:full_non_epis}). Hence $\Lambda_C\cap\left(\bigcap_{i=1}^n\Lambda_{\vec{v_i}}\right)\neq \varnothing$ and,
\begin{equation}
\prod_{i=1}^n\Theta(\vec{v}_i.\vec{r}_1)\Theta(\vec{v}_i.\vec{r}_2)=1 \textrm{ } \forall \textrm{ } \vec{r}_1, \vec{r}_2 \in \Lambda_C\cap\left(\bigcap_{i=1}^n\Lambda_{\vec{v_i}}\right),
\end{equation}
which implies that for $\lambda=(c_1=1,c_2=1;\vec{r}_1,\vec{r}_2)$,
\begin{equation}
\min(\mu(\lambda|\vec{v_1}),\cdots,\mu(\lambda|\vec{v_n}))>0 \textrm{ }  \forall \vec{r}_1, \vec{r}_2  \in \Lambda_C\cap\left(\bigcap_{i=1}^n\Lambda_{\vec{v_i}}\right),
\end{equation}
which leads to $\omega_E^{[n]}(\psi_1,\cdots,\psi_n)>0$. \end{proof}

\begin{figure}[h]
        \centering
        \includegraphics[width=0.38\textwidth]{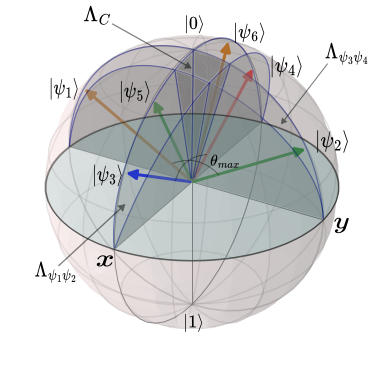}
        \caption{(Lemma \ref{lemm_KS}) Six pure qubit states are located in the hemisphere $z>0$, which is denoted as '$I$' in the proof. Out of them, $\ket{\psi_1}$ and $\ket{\psi_2}$ are situated in the $Z-Y$ plane while $\ket{\psi_3}$ and $\ket{\psi_4}$ are situated in the $Z-X$ plane. $\ket{\psi_1}, \ket{\psi_2}$ and $\ket{\psi_3}, \ket{\psi_4}$ have the largest angle, that is, $\theta(1,2)=\theta(3,4)=\theta_{max}$. $\Lambda_{\psi_1\psi_2}$ denotes the $\mathcal{S}^2$ subspace for which $\Theta(\vec{r}_1.\vec{v}_{\psi_1})\Theta(\vec{r}_2.\vec{v}_{\psi_1})\Theta(\vec{r}_1.\vec{v}_{\psi_2})\Theta(\vec{r}_2.\vec{v}_{\psi_2})=1$ for all $\vec{r}_1,\vec{r}_2 \in \Lambda_{\psi_1\psi_2}$ and likewise for $\Lambda_{\psi_3\psi_4}$. $\Lambda_C$ is the intersection of $\Lambda_{\psi_1\psi_2}$ and $\Lambda_{\psi_3\psi_4}$ and hence, $\prod_{i=1}^4\Theta(\vec{v}_{\psi_{i}}.\vec{r}_1)\Theta(\vec{v}_{\psi_i}.\vec{r}_2)=1$ $\forall \vec{r}_1, \vec{r}_2 \in \Lambda_C$. Note that the region $R$, as mentioned in the proof, has not been shown here to prevent overcrowding in the figure. It can be identified by looking at the tips of $\ket{\psi_1}, \ket{\psi_3}, \ket{\psi_2}$ and $\ket{\psi_4}$ since they lie on the region's circumference. It is evident that $\ket{\psi_5}$ and $\ket{\psi_6}$ lie in this region. If we take $c_1=c_2=1$ then $\Lambda_C$ is a subset of the overlap of $\ket{\psi_5}$ and $\ket{\psi_6}$'s supports and hence $\prod_{i=1}^6\mu(\lambda|\psi_i)>0$ $\forall \vec{r}_1, \vec{r}_2 \in \Lambda_C$ and $c_1=c_2=1$, which implies non-zero epistemic overlap.}
        \label{fig:full_non_epis}
    \end{figure}
This result can be considered as the RTQ model version of Lemma \ref{lem_bloch}. Recall that Lemma \ref{lem_bloch} describes how qubit states $\{\ket{\psi_i}\}_{i=1}^n$ belonging to the same hemisphere of the Bloch sphere are not perfectly anti-distinguishable, meaning that $\omega_Q^{[n]}(\psi_1,\cdots,\psi_n)>0$. Similarly Lemma \ref{lemm_KS} shows that in RTQ model $\omega_E^{[n]}(\psi_1,\cdots,\psi_n)>0$. In a way, Lemma \ref{lem_bloch} follows from Lemma \ref{lemm_KS} since, once we have shown that $\omega_E^{[n]}(\psi_1,\cdots,\psi_n)>0$, from \eqref{we<wq} $\omega_Q^{[n]}(\psi_1,\cdots,\psi_n)$ has to be greater than zero as well.

  \begin{proof}[Proof of Theorem \ref{thm_no_fully_epis}]
    Let us consider that we have a set of $n$ mixed qubit preparations $\{\rho_k\}_{k=1}^n$, of the same mixed state $\varrho$ whose Bloch vector is $\vec{v}$, $|\vec{v}|<1$. Each one is a convex mixture of pure states $\{\ket{\psi_{i_k|k}}\}_{i_k=1}^m$ with convex coefficients $\{p_{i_k|k}\}_{i_k=1}^m$ such that $\rho_k=\sum_{i_k=1}^mp_{i_k|k}\ket{\psi_{i_k|k}}\!\langle\psi_{i|k}|$. The epistemic overlap $\omega_E^{[n]}(\rho_1,\cdots,\rho_n)$ has an upper-bound as described in \eqref{sub_eq1}. Let us divide the Bloch sphere into two hemispheres, labeled $I$ and $II$, and let $\vec{v}$ be in $I$. Due to convexity, at least one pure state from the decomposition of each mixed preparation will lie in hemisphere $I$. Thus, there exists at least one such set $\{\ket{\psi_{i_1|1}},\ket{\psi_{i_2|2}},\cdots,\ket{\psi_{i_n|n}}\}$ which lies in $I$. According to Theorem \ref{t1}, if $\omega_E^{[n]}(\rho_1,\cdots,\rho_n)=0$, \eqref{sub_eq1} becomes an equality, and every $\omega_E^{[n]}(\psi_{i_1|1},\cdots,\psi_{i_n|n})=0$ for all values of $\{i_1,\cdots,i_n\}$. However, since we have at least one set  $\{\ket{\psi_{i_1|1}},\ket{\psi_{i_2|2}},\cdots,\ket{\psi_{i_n|n}}\}$ which lies in hemisphere $I$, according to Lemma \ref{lemm_KS}, there exists an epistemic model, which is the RTQ model where the epistemic overlap of that set $\omega_E^{[n]}(\psi_{i_1|1},\cdots,\psi_{i_n|n})>0$. Thus, $\omega_E^{[n]}(\rho_1,\cdots,\rho_n)>0$ as well. 

Now suppose $\vec{v}$ belongs to neither hemisphere of the Bloch sphere, whichever way we may divide the sphere. In that case, it necessarily lies at the origin, meaning we have a set of $n$ maximally mixed qubit preparations. The preparations are of the form, $\rho_k=1/2(\ket{\psi_k}\!\la\psi_k|+\ket{\overline{\psi_k}}\!\la\overline{\psi_k|})$, where $\ket{\psi_k}$ and $\ket{\overline{\psi_k}}$ are orthogonal. The epistemic state corresponding to $\rho_k$ is,
    \begin{equation}
        \mu(\lambda|\rho_k)=\frac{1}{2}\left(\mu(\lambda|\psi_k)+\mu(\lambda|\overline{\psi_k})\right).
    \end{equation}
    In RTQ model, the supports of $\mu(\lambda|\psi_k)$ and $\mu(\lambda|\overline{\psi_k})$, where $\lambda=(c_1=1,c_2=1;\vec{r}_1,\vec{r}_2)$ namely $\Lambda_{\psi_k}\subset \mathcal{S}^2$ and $\Lambda_{\overline{\psi_k}}\subset \mathcal{S}^2$ corresponds to two hemispheres with no overlap. In other words, the support of $\mu(\lambda|\rho_k)$ for $c_1=c_2=1$ is the entire $\mathcal{S}^2-$space. Consequently there exists some $\vec{r}_1, \vec{r}_2 \in \mathcal{S}^2$ such that $\prod_{k=1}^n\mu(\lambda|\rho_k)>0$, implying that $\omega_E^{[n]}(\rho_1,\cdots,\rho_n)>0$.
    
    Thus, no matter what mixed qubit state we have, whether it be maximally mixed or not, in the RTQ model, its decompositions have a non-zero epistemic overlap, implying that fully non-epistemic case is not possible.
    \end{proof}

\bibliography{ref} 

@article{Barrett,
  title = {No $\ensuremath{\psi}$-Epistemic Model Can Fully Explain the Indistinguishability of Quantum States},
  author = {Barrett, Jonathan and Cavalcanti, Eric G. and Lal, Raymond and Maroney, Owen J. E.},
  journal = {Phys. Rev. Lett.},
  volume = {112},
  issue = {25},
  pages = {250403},
  numpages = {6},
  year = {2014},
  month = {Jun},
  publisher = {American Physical Society},
  doi = {10.1103/PhysRevLett.112.250403},
  url = {https://link.aps.org/doi/10.1103/PhysRevLett.112.250403}
}

@article{Leifer,
  title = {$\ensuremath{\psi}$-Epistemic Models are Exponentially Bad at Explaining the Distinguishability of Quantum States},
  author = {Leifer, M. S.},
  journal = {Phys. Rev. Lett.},
  volume = {112},
  issue = {16},
  pages = {160404},
  numpages = {4},
  year = {2014},
  month = {Apr},
  publisher = {American Physical Society},
  doi = {10.1103/PhysRevLett.112.160404},
  url = {https://link.aps.org/doi/10.1103/PhysRevLett.112.160404}
}

@article{cs,
  doi = {10.22331/q-2020-10-21-345},
  url = {https://doi.org/10.22331/q-2020-10-21-345},
  title = {Quantum prescriptions are more ontologically distinct than they are operationally distinguishable},
  author = {Chaturvedi, Anubhav and Saha, Debashis},
  journal = {{Quantum}},
  issn = {2521-327X},
  publisher = {{Verein zur F{\"{o}}rderung des Open Access Publizierens in den Quantenwissenschaften}},
  volume = {4},
  pages = {345},
  month = oct,
  year = {2020}
}

@article{Branciard,
  title = {How $\ensuremath{\psi}$-Epistemic Models Fail at Explaining the Indistinguishability of Quantum States},
  author = {Branciard, Cyril},
  journal = {Phys. Rev. Lett.},
  volume = {113},
  issue = {2},
  pages = {020409},
  numpages = {5},
  year = {2014},
  month = {Jul},
  publisher = {American Physical Society},
  doi = {10.1103/PhysRevLett.113.020409},
  url = {https://link.aps.org/doi/10.1103/PhysRevLett.113.020409}
}

@article{caves,
  title = {Conditions for compatibility of quantum-state assignments},
  author = {Caves, Carlton M. and Fuchs, Christopher A. and Schack, R\"udiger},
  journal = {Phys. Rev. A},
  volume = {66},
  issue = {6},
  pages = {062111},
  numpages = {11},
  year = {2002},
  month = {Dec},
  publisher = {American Physical Society},
  doi = {10.1103/PhysRevA.66.062111},
  url = {https://link.aps.org/doi/10.1103/PhysRevA.66.062111}
}

@article{johnston2023tight,
  doi = {10.22331/q-2025-02-04-1622},
  url = {https://doi.org/10.22331/q-2025-02-04-1622},
  title = {Tight bounds for antidistinguishability and circulant sets of pure quantum states},
  author = {Johnston, Nathaniel and Russo, Vincent and Sikora, Jamie},
  journal = {{Quantum}},
  issn = {2521-327X},
  publisher = {{Verein zur F{\"{o}}rderung des Open Access Publizierens in den Quantenwissenschaften}},
  volume = {9},
  pages = {1622},
  month = feb,
  year = {2025}
}

@article{PuseyPRA2018,
  title = {Robust preparation noncontextuality inequalities in the simplest scenario},
  author = {Pusey, Matthew F.},
  journal = {Phys. Rev. A},
  volume = {98},
  issue = {2},
  pages = {022112},
  numpages = {8},
  year = {2018},
  month = {Aug},
  publisher = {American Physical Society},
  doi = {10.1103/PhysRevA.98.022112},
  url = {https://link.aps.org/doi/10.1103/PhysRevA.98.022112}
}

@article{PBR,
  title={On the reality of the quantum state},
  author={Pusey, Matthew F and Barrett, Jonathan and Rudolph, Terry},
  journal={Nature Physics},
  volume={8},
  number={6},
  pages={475--478},
  year={2012},
  publisher={Nature Publishing Group UK London}
}

@article{Lewis,
  title = {Distinct Quantum States Can Be Compatible with a Single State of Reality},
  author = {Lewis, Peter G. and Jennings, David and Barrett, Jonathan and Rudolph, Terry},
  journal = {Phys. Rev. Lett.},
  volume = {109},
  issue = {15},
  pages = {150404},
  numpages = {5},
  year = {2012},
  month = {Oct},
  publisher = {American Physical Society},
  doi = {10.1103/PhysRevLett.109.150404},
  url = {https://link.aps.org/doi/10.1103/PhysRevLett.109.150404}
}

@article{Aaronson,
  title = {$\ensuremath{\psi}$-epistemic theories: The role of symmetry},
  author = {Aaronson, Scott and Bouland, Adam and Chua, Lynn and Lowther, George},
  journal = {Phys. Rev. A},
  volume = {88},
  issue = {3},
  pages = {032111},
  numpages = {12},
  year = {2013},
  month = {Sep},
  publisher = {American Physical Society},
  doi = {10.1103/PhysRevA.88.032111},
  url = {https://link.aps.org/doi/10.1103/PhysRevA.88.032111}
}

@article{harrigan2010einstein,
  title={Einstein, incompleteness, and the epistemic view of quantum states},
  author={Harrigan, Nicholas and Spekkens, Robert W},
  journal={Foundations of Physics},
  volume={40},
  pages={125--157},
  year={2010},
  publisher={Springer}
}

@article{Spekkens-toy,
  title = {Evidence for the epistemic view of quantum states: A toy theory},
  author = {Spekkens, Robert W.},
  journal = {Phys. Rev. A},
  volume = {75},
  issue = {3},
  pages = {032110},
  numpages = {30},
  year = {2007},
  month = {Mar},
  publisher = {American Physical Society},
  doi = {10.1103/PhysRevA.75.032110},
  url = {https://link.aps.org/doi/10.1103/PhysRevA.75.032110}
}

@article{Spekkens2005,
  title = {Contextuality for preparations, transformations, and unsharp measurements},
  author = {Spekkens, R. W.},
  journal = {Phys. Rev. A},
  volume = {71},
  issue = {5},
  pages = {052108},
  numpages = {17},
  year = {2005},
  month = {May},
  publisher = {American Physical Society},
  doi = {10.1103/PhysRevA.71.052108},
  url = {https://link.aps.org/doi/10.1103/PhysRevA.71.052108}
}

@article{Spekkens2009,
  title = {Preparation Contextuality Powers Parity-Oblivious Multiplexing},
  author = {Spekkens, Robert W. and Buzacott, D. H. and Keehn, A. J. and Toner, Ben and Pryde, G. J.},
  journal = {Phys. Rev. Lett.},
  volume = {102},
  issue = {1},
  pages = {010401},
  numpages = {4},
  year = {2009},
  month = {Jan},
  publisher = {American Physical Society},
  doi = {10.1103/PhysRevLett.102.010401},
  url = {https://link.aps.org/doi/10.1103/PhysRevLett.102.010401}
}

@article{Leifer-review,
  title={Is the Quantum State Real? An Extended Review of $\psi$-ontology Theorems},
  author={Leifer, Matthew S},
  journal={Quanta},
  volume={3},
  number={1},
  pages={67--155},
  year={2014},
  doi = {10.12743/quanta.v3i1.22}
}

@article{kochen1967problem,
 ISSN = {00959057, 19435274},
 URL = {http://www.jstor.org/stable/24902153},
 author = {Simon Kochen and E. P. Specker},
 journal = {Journal of Mathematics and Mechanics},
 number = {1},
 pages = {59--87},
 publisher = {Indiana University Mathematics Department},
 title = {The Problem of Hidden Variables in Quantum Mechanics},
 volume = {17},
 year = {1967},
 doi = {10.2307/24902153}
}

@article{RTQ,
  title = {Classical Cost of Transmitting a Qubit},
  author = {Renner, Martin J. and Tavakoli, Armin and Quintino, Marco T\'ulio},
  journal = {Phys. Rev. Lett.},
  volume = {130},
  issue = {12},
  pages = {120801},
  numpages = {7},
  year = {2023},
  month = {Mar},
  publisher = {American Physical Society},
  doi = {10.1103/PhysRevLett.130.120801},
  url = {https://link.aps.org/doi/10.1103/PhysRevLett.130.120801}
}

@article{hance2022wave,
  title={The wave function as a true ensemble},
  author={Hance, Jonte R and Hossenfelder, Sabine},
  journal={Proceedings of the Royal Society A},
  volume={478},
  number={2262},
  pages={20210705},
  year={2022},
  publisher={The Royal Society},
url = {https://royalsocietypublishing.org/doi/10.1098/rspa.2021.0705}
}

@article{maroney2012statistical,
  title={How statistical are quantum states?},
  author={Maroney, Owen JE},
  journal={arXiv preprint arXiv:1207.6906},
  year={2012}
}

@article{Leifer-Maroney,
  title = {Maximally Epistemic Interpretations of the Quantum State and Contextuality},
  author = {Leifer, M. S. and Maroney, O. J. E.},
  journal = {Phys. Rev. Lett.},
  volume = {110},
  issue = {12},
  pages = {120401},
  numpages = {5},
  year = {2013},
  month = {Mar},
  publisher = {American Physical Society},
  doi = {10.1103/PhysRevLett.110.120401},
  url = {https://link.aps.org/doi/10.1103/PhysRevLett.110.120401}
}

@article{Renner,
  title = {Is a System's Wave Function in One-to-One Correspondence with Its Elements of Reality?},
  author = {Colbeck, Roger and Renner, Renato},
  journal = {Phys. Rev. Lett.},
  volume = {108},
  issue = {15},
  pages = {150402},
  numpages = {4},
  year = {2012},
  month = {Apr},
  publisher = {American Physical Society},
  doi = {10.1103/PhysRevLett.108.150402},
  url = {https://link.aps.org/doi/10.1103/PhysRevLett.108.150402}
}

@article{app1,
  title = {Communication Tasks with Infinite Quantum-Classical Separation},
  author = {Perry, Christopher and Jain, Rahul and Oppenheim, Jonathan},
  journal = {Phys. Rev. Lett.},
  volume = {115},
  issue = {3},
  pages = {030504},
  numpages = {5},
  year = {2015},
  month = {Jul},
  publisher = {American Physical Society},
  doi = {10.1103/PhysRevLett.115.030504},
  url = {https://link.aps.org/doi/10.1103/PhysRevLett.115.030504}
}

@article{app3,
  title = {Epistemic View of Quantum States and Communication Complexity of Quantum Channels},
  author = {Montina, Alberto},
  journal = {Phys. Rev. Lett.},
  volume = {109},
  issue = {11},
  pages = {110501},
  numpages = {4},
  year = {2012},
  month = {Sep},
  publisher = {American Physical Society},
  doi = {10.1103/PhysRevLett.109.110501},
  url = {https://link.aps.org/doi/10.1103/PhysRevLett.109.110501}
}

@article{app2,
  title = {Simple communication complexity separation from quantum state antidistinguishability},
  author = {Havl\'{\i}\ifmmode \check{c}\else \v{c}\fi{}ek, Vojtech and Barrett, Jonathan},
  journal = {Phys. Rev. Res.},
  volume = {2},
  issue = {1},
  pages = {013326},
  numpages = {8},
  year = {2020},
  month = {Mar},
  publisher = {American Physical Society},
  doi = {10.1103/PhysRevResearch.2.013326},
  url = {https://link.aps.org/doi/10.1103/PhysRevResearch.2.013326}
}

@article{sahaPRA,
  title = {Preparation contextuality as an essential feature underlying quantum communication advantage},
  author = {Saha, Debashis and Chaturvedi, Anubhav},
  journal = {Phys. Rev. A},
  volume = {100},
  issue = {2},
  pages = {022108},
  numpages = {13},
  year = {2019},
  month = {Aug},
  publisher = {American Physical Society},
  doi = {10.1103/PhysRevA.100.022108},
  url = {https://link.aps.org/doi/10.1103/PhysRevA.100.022108}
}

@article{SahaNJP,
	doi = {10.1088/1367-2630/ab4149},
	url = {https://doi.org/10.1088/1367-2630/ab4149},
	year = 2019,
	month = {sep},
	publisher = {{IOP} Publishing},
	volume = {21},
	number = {9},
	pages = {093057},
	author = {Debashis Saha and Pawe{\l} Horodecki and Marcin Paw{\l}owski},
	title = {State independent contextuality advances one-way communication},
	journal = {New J. Phys.}
}

@article{XuPRA,
  title = {Reformulating noncontextuality inequalities in an operational approach},
  author = {Xu, Zhen-Peng and Saha, Debashis and Su, Hong-Yi and Paw\l{}owski, Marcin and Chen, Jing-Ling},
  journal = {Phys. Rev. A},
  volume = {94},
  issue = {6},
  pages = {062103},
  numpages = {6},
  year = {2016},
  month = {Dec},
  publisher = {American Physical Society},
  doi = {10.1103/PhysRevA.94.062103},
  url = {https://link.aps.org/doi/10.1103/PhysRevA.94.062103}
}

@article{PanPRA,
  title = {Optimal quantum preparation contextuality in an $n$-bit parity-oblivious multiplexing task},
  author = {Ghorai, Shouvik and Pan, A. K.},
  journal = {Phys. Rev. A},
  volume = {98},
  issue = {3},
  pages = {032110},
  numpages = {8},
  year = {2018},
  month = {Sep},
  publisher = {American Physical Society},
  doi = {10.1103/PhysRevA.98.032110},
  url = {https://link.aps.org/doi/10.1103/PhysRevA.98.032110}
}

@article{SikoraNJP,
doi = {10.1088/1367-2630/18/4/045003},
url = {https://dx.doi.org/10.1088/1367-2630/18/4/045003},
year = {2016},
month = {apr},
publisher = {IOP Publishing},
volume = {18},
number = {4},
pages = {045003},
author = {André Chailloux and Iordanis Kerenidis and Srijita Kundu and Jamie Sikora},
title = {Optimal bounds for parity-oblivious random access codes},
journal = {New Journal of Physics}
}

@article{interference,
  doi = {10.22331/q-2023-09-25-1119},
  url = {https://doi.org/10.22331/q-2023-09-25-1119},
  title = {Why interference phenomena do not capture the essence of quantum theory},
  author = {Catani, Lorenzo and Leifer, Matthew and Schmid, David and Spekkens, Robert W.},
  journal = {{Quantum}},
  issn = {2521-327X},
  publisher = {{Verein zur F{\"{o}}rderung des Open Access Publizierens in den Quantenwissenschaften}},
  volume = {7},
  pages = {1119},
  month = sep,
  year = {2023}
}

@article{SchmidPRL22,
  title = {Uniqueness of Noncontextual Models for Stabilizer Subtheories},
  author = {Schmid, David and Du, Haoxing and Selby, John H. and Pusey, Matthew F.},
  journal = {Phys. Rev. Lett.},
  volume = {129},
  issue = {12},
  pages = {120403},
  numpages = {6},
  year = {2022},
  month = {Sep},
  publisher = {American Physical Society},
  doi = {10.1103/PhysRevLett.129.120403},
  url = {https://link.aps.org/doi/10.1103/PhysRevLett.129.120403}
}

@article{SchmidPRX,
  title = {Contextual Advantage for State Discrimination},
  author = {Schmid, David and Spekkens, Robert W.},
  journal = {Phys. Rev. X},
  volume = {8},
  issue = {1},
  pages = {011015},
  numpages = {20},
  year = {2018},
  month = {Feb},
  publisher = {American Physical Society},
  doi = {10.1103/PhysRevX.8.011015},
  url = {https://link.aps.org/doi/10.1103/PhysRevX.8.011015}
}

@article{Review,
title = {Specker’s parable of the overprotective seer: A road to contextuality, nonlocality and complementarity},
journal = {Physics Reports},
volume = {506},
number = {1},
pages = {1-39},
year = {2011},
issn = {0370-1573},
doi = {https://doi.org/10.1016/j.physrep.2011.05.001},
url = {https://www.sciencedirect.com/science/article/pii/S0370157311001517},
author = {Yeong-Cherng Liang and Robert W. Spekkens and Howard M. Wiseman}
}

@article{LostaglioPRL,
  title = {Quantum Fluctuation Theorems, Contextuality, and Work Quasiprobabilities},
  author = {Lostaglio, Matteo},
  journal = {Phys. Rev. Lett.},
  volume = {120},
  issue = {4},
  pages = {040602},
  numpages = {6},
  year = {2018},
  month = {Jan},
  publisher = {American Physical Society},
  doi = {10.1103/PhysRevLett.120.040602},
  url = {https://link.aps.org/doi/10.1103/PhysRevLett.120.040602}
}

@article{Spekkens2008prl,
  title = {Negativity and Contextuality are Equivalent Notions of Nonclassicality},
  author = {Spekkens, Robert W.},
  journal = {Phys. Rev. Lett.},
  volume = {101},
  issue = {2},
  pages = {020401},
  numpages = {4},
  year = {2008},
  month = {Jul},
  publisher = {American Physical Society},
  doi = {10.1103/PhysRevLett.101.020401},
  url = {https://link.aps.org/doi/10.1103/PhysRevLett.101.020401}
}

@article{PuseyPRL,
  title = {Anomalous Weak Values Are Proofs of Contextuality},
  author = {Pusey, Matthew F.},
  journal = {Phys. Rev. Lett.},
  volume = {113},
  issue = {20},
  pages = {200401},
  numpages = {5},
  year = {2014},
  month = {Nov},
  publisher = {American Physical Society},
  doi = {10.1103/PhysRevLett.113.200401},
  url = {https://link.aps.org/doi/10.1103/PhysRevLett.113.200401}
}

@article{mate-prl-2023,
  title = {Invertible Map between Bell Nonlocal and Contextuality Scenarios},
  author = {Wright, Victoria J. and Farkas, M\'at\'e},
  journal = {Phys. Rev. Lett.},
  volume = {131},
  issue = {22},
  pages = {220202},
  numpages = {6},
  year = {2023},
  month = {Nov},
  publisher = {American Physical Society},
  doi = {10.1103/PhysRevLett.131.220202},
  url = {https://link.aps.org/doi/10.1103/PhysRevLett.131.220202}
}

@article{UR,
  title = {What is Nonclassical about Uncertainty Relations?},
  author = {Catani, Lorenzo and Leifer, Matthew and Scala, Giovanni and Schmid, David and Spekkens, Robert W.},
  journal = {Phys. Rev. Lett.},
  volume = {129},
  issue = {24},
  pages = {240401},
  numpages = {6},
  year = {2022},
  month = {Dec},
  publisher = {American Physical Society},
  doi = {10.1103/PhysRevLett.129.240401},
  url = {https://link.aps.org/doi/10.1103/PhysRevLett.129.240401}
}

@article{Randomness_certify,
  title = {Quantum vs Noncontextual Semi-Device-Independent Randomness Certification},
  author = {Roch i Carceller, Carles and Flatt, Kieran and Lee, Hanwool and Bae, Joonwoo and Brask, Jonatan Bohr},
  journal = {Phys. Rev. Lett.},
  volume = {129},
  issue = {5},
  pages = {050501},
  numpages = {6},
  year = {2022},
  month = {Jul},
  publisher = {American Physical Society},
  doi = {10.1103/PhysRevLett.129.050501},
  url = {https://link.aps.org/doi/10.1103/PhysRevLett.129.050501}
}

@article{State_discri_2,
  title = {Contextual Advantages and Certification for Maximum-Confidence Discrimination},
  author = {Flatt, Kieran and Lee, Hanwool and Carceller, Carles Roch I and Brask, Jonatan Bohr and Bae, Joonwoo},
  journal = {PRX Quantum},
  volume = {3},
  issue = {3},
  pages = {030337},
  numpages = {19},
  year = {2022},
  month = {Sep},
  publisher = {American Physical Society},
  doi = {10.1103/PRXQuantum.3.030337},
  url = {https://link.aps.org/doi/10.1103/PRXQuantum.3.030337}
}

@article{Chaturvedi2021characterising,
  doi = {10.22331/q-2021-06-29-484},
  url = {https://doi.org/10.22331/q-2021-06-29-484},
  title = {Characterising and bounding the set of quantum behaviours in contextuality scenarios},
  author = {Chaturvedi, Anubhav and Farkas, M{\'{a}}t{\'{e}} and Wright, Victoria J},
  journal = {{Quantum}},
  issn = {2521-327X},
  publisher = {{Verein zur F{\"{o}}rderung des Open Access Publizierens in den Quantenwissenschaften}},
  volume = {5},
  pages = {484},
  month = jun,
  year = {2021}
}

@article{Plavala-prl,
  title = {Contextuality as a Precondition for Quantum Entanglement},
  author = {Pl\'avala, Martin and G\"uhne, Otfried},
  journal = {Phys. Rev. Lett.},
  volume = {132},
  issue = {10},
  pages = {100201},
  numpages = {7},
  year = {2024},
  month = {Mar},
  publisher = {American Physical Society},
  doi = {10.1103/PhysRevLett.132.100201},
  url = {https://link.aps.org/doi/10.1103/PhysRevLett.132.100201}
}

@article{Shin2021QuantumContextual,
  author  = {Shin, J. and Ha, D. and Kwon, Y.},
  title   = {Quantum Contextual Advantage Depending on Nonzero Prior Probabilities in State Discrimination of Mixed Qubit States},
  journal = {Entropy},
  year    = {2021},
  volume  = {23},
  number  = {12},
  pages   = {1583},
  doi     = {10.3390/e23121583},
}

\end{document}